%% file: 0_main.tex
\documentclass{amsart}

\usepackage{amscd}
\usepackage{bbm}
\usepackage{mathrsfs}
\usepackage{eucal}
\usepackage{faktor}
\usepackage{xfrac}
\usepackage{braket}
\usepackage{relsize} 
\usepackage{nccmath}
\usepackage{mathtools}
\usepackage{comment}
\usepackage{verbatim}
\usepackage{wasysym}
\usepackage{graphicx}
\usepackage{caption}
\usepackage{import}
\usepackage{latexsym}
\usepackage{amsfonts}
\usepackage{amssymb}
\usepackage{amsmath}
\usepackage{mathtext}
\usepackage{cite}
\usepackage{enumerate}
\usepackage{float}
\usepackage{amsthm}
\usepackage{upgreek}
\usepackage{nicefrac}
\usepackage{multicol}
\usepackage{rotating}
\usepackage{array}

\usepackage{etoolbox}
\let\oldtocsection=\tocsection
\let\oldtocsubsection=\tocsubsection
\let\oldtocsubsubsection=\tocsubsubsection
\renewcommand{\tocsection}[2]{\hspace{0em}\oldtocsection{#1}{#2}}
\renewcommand{\tocsubsection}[2]{\hspace{1.8em}\oldtocsubsection{#1}{#2}}
\renewcommand{\tocsubsubsection}[2]{\hspace{4.4em}\oldtocsubsubsection{#1}{#2}}

\usepackage[svgnames]{xcolor}
\definecolor{linkcolor}{HTML}{e88d67} 
\definecolor{citecolor}{HTML}{e88d67} 
\definecolor{urlcolor}{HTML}{e88d67} 
\definecolor{myNewColorA}{HTML}{fec3a6}
\definecolor{myNewColorB}{HTML}{ffaf80}
\definecolor{myNewColorC}{HTML}{fb8f67}
\definecolor{seagreen}{HTML}{337180}
\definecolor{mseagreen}{HTML}{369673}
\definecolor{darksalmon}{HTML}{e88d67}
\definecolor{silver}{HTML}{bbbbbb}
\definecolor{flowerblue}{HTML}{4e77fc}
\definecolor{tomato}{HTML}{ff6347}
\definecolor{orange}{HTML}{f2b13d}
\definecolor{darkgray}{HTML}{939393}

\usepackage[
  bookmarks = true,
  bookmarksopen = true,
  colorlinks = , 
  urlcolor = urlcolor, 
  linkcolor = linkcolor, 
  citecolor = citecolor, 
  linktoc = all,
  pagebackref = true
]{hyperref} 
\mathtoolsset{showonlyrefs,showmanualtags}

\DeclareMathOperator{\ad}{ad}
\DeclareMathOperator{\Mat}{Mat}

\newcommand{\brackets}[1]{\left( #1 \right)}
\newcommand{\PoissonBrackets}[1]{\left\{ #1 \right\}}
\newcommand{\LieBrackets}[1]{\left[ #1 \right]}
\newcommand{\angleBrackets}[1]{\langle #1 \rangle}

\newcommand{\dsum}{\displaystyle \sum}
\newcommand{\quot}[2]{{\raisebox{.3em}{$#1\!$}\bigl/\raisebox{-.3em}{$\!#2$}}}

\newcommand{\Painleve}{Painlev{\'e} }
\newcommand{\PPainleve}{Painlev{\'e}}

\newcommand{\tomato}[1]{{\color{tomato} #1}}

\theoremstyle{plain}
\newtheorem{thm}{Theorem}[]

\newtheorem{lem}{Lemma}[]

\newtheorem{prop}{Proposition}[]

\theoremstyle{definition}
\newtheorem{defn}{Definition}[]

\theoremstyle{remark}
\newtheorem{rem}{Remark}

\usepackage{chngcntr}
\counterwithin*{thm}{section} 
\counterwithin*{lem}{section} 
\counterwithin*{claim}{section} 
\counterwithin*{prop}{section} 
\counterwithin*{corl}{section} 
\counterwithin*{defn}{section} 
\counterwithin*{exmp}{section} 
\counterwithin*{ex}{section} 
\counterwithin*{rem}{section} 

\usepackage{apptools}
\AtAppendix{\counterwithin{thm}{section}}
\AtAppendix{\counterwithin{lem}{section}}
\AtAppendix{\counterwithin{claim}{section}}
\AtAppendix{\counterwithin{prop}{section}}
\AtAppendix{\counterwithin{corl}{section}}
\AtAppendix{\counterwithin{defn}{section}}
\AtAppendix{\counterwithin{exmp}{section}}
\AtAppendix{\counterwithin{ex}{section}}
\AtAppendix{\counterwithin{rem}{section}}

\usepackage{tikz}

\begin{document}

    \title[Non-abelian $\rm{P}_2$-systems: linearizations and monodromy surfaces]{Different linearizations of~non-abelian second Painlev\'e systems and related monodromy surfaces}

    \author{Irina Bobrova}
    \noindent\address{\noindent 
    Faculty of Mathematics, HSE University, Usacheva str. 6, Moscow, Russia}
    \email{ia.bobrova94@gmail.com}

    
    \maketitle

    \begin{flushright}
    \textit{
    Dedicated to the 75th birthday of V.S. Retakh
    }
    \end{flushright}
    
    \begin{abstract}
    In this paper, we discuss a connection between different linearizations for non-abelian analogs of the second \Painleve equation. For each of the analogs, we listed the pairs of the Harnard-Tracy-Widom (HTW), Flaschka-Newell (FN), and Jimbo-Miwa (JM) types. A method for establishing the HTW-JM correspondence is suggested. For one of the non-abelian analogs, we derive the corresponding non-abelian generalizations of the monodromy surfaces related to the FN- and JM-type linearizations. A natural Poisson structure associated with these monodromy surfaces is also discussed. 
    \end{abstract}
    
    \tableofcontents
    
    \input{1_body.tex}

    \medskip 
    \subsubsection*{Availability of data and materials}
    Not applicable.
    
    \subsubsection*{Competing interests}
    The authors declare that they have no competing interests.

    \subsubsection*{Funding}
    International Laboratory of Cluster Geometry HSE, RF Government grant № 075-15-2021-608; 
    Young Russian Mathematics award.

    \bibliographystyle{alpha}
    \bibliography{bib}
\end{document}

%% file: 1_body.tex
\section{Introduction}

The famous differential \Painleve equations were derived more than one hundred years ago \cite{painleve1900memoire}, \cite{painleve1902equations}, \cite{gambier1910equations}. These equations have being studied in various branches of mathematics and mathematical physics and have important properties. In particular, they are related to a system of scalar differential equations \cite{fuchs1907lineare}, \cite{garnier1912equations} integrable in the sense of the Frobenius theorem. In the paper \cite{jimbo1981monodromy}, it was established that the \Painleve equations can be linearized. This fact is connected with monodromy preserving deformations related to vector bundles of rank 2. Due to the isomonodromic property, the space of solutions of the \Painleve equations can be parameterized by the monodromy data. Namely, each of the equations can be associated with the zero-locus of an affine cubic that is usually called \textit{the~monodromy~surface} (e.g., \cite{van2009moduli}). Informally speaking, the dimension of the space of solutions is equal to the dimension of the monodromy surface.

In this paper we are interested in non-abelian generalizations of the well-known monodromy surfaces related to different linearizations of the non-commutative analogs for the second \Painleve equation, obtained in~\cite{Adler_Sokolov_2020_1}. The commutative second \Painleve equation is given by
\begin{align}
    &&
    f''
    &= 2 f^3 + x f + \brackets{\theta - \tfrac12},
    &
    f(x), \,\, x, \,\, \theta 
    &\in \mathbb{C},
    &&
\end{align}
and can be rewritten as the system \cite{okamoto1980polynomial}
\begin{align}
    &&
    \label{eq:P2sys}
    \tag*{$\rm{P}_2$}
    &\left\{
    \begin{array}{lcl}
         f'
         &=& - f^2 + g - \frac12 x,  
         \\[2mm]
         g'
         &=& 2 f g + \theta,
    \end{array}
    \right.
    &
    f(x), \,\, g(x), \,\, x, \,\, \theta
    &\in \mathbb{C}.
    &&
\end{align}
This system admits several types of the isomonodromic Lax pairs (see \cite{flaschka1980monodromy}, \cite{jimbo1981monodromy}, \cite{harnad1993hamiltonian}), i.e. it is equivalent to the compatibility condition of the following linear system
\begin{align}
    \label{eq:linsys}
    \left\{
    \begin{array}{lcl}
         \partial_{\lambda} \Phi (\lambda, x)
         &=& A (\lambda, x) \, \Phi (\lambda, x)
         ,  
         \\[3mm]
         \partial_{x} \Phi (\lambda, x)
         &=& B (\lambda, x) \, \Phi (\lambda, x)
         ,
    \end{array}
    \right.
\end{align}
where $\lambda$ is the spectral parameter and the matrix $A (\lambda, x)$ may have different dependence on $\lambda$. The matrix $A (\lambda, x)$ is classified into three types: 
\begin{itemize}
    \item 
    the Harnad-Tracy-Widom (HTW) type \cite{harnad1993hamiltonian}
    \begin{align}
        \label{eq:HTWpair_form}
        A (\lambda, x)
        &= A_{1} \lambda + A_0 (x) + A_{-1} (x) \lambda^{-1},
    \end{align}
    
    \item
    the Flaschka-Newell (FN) type \cite{flaschka1980monodromy}
    \begin{align}
        \label{eq:FNpair_form}
        A (\lambda, x)
        &= A_2 \lambda^2 
        + A_{1} (x) \lambda 
        + A_0 (x) + A_{-1} (x) \lambda^{-1},
    \end{align}
    
    \item
    and the Jimbo-Miwa (JM) type \cite{jimbo1981monodromy}
    \begin{align}
        \label{eq:JMpair_form}
        A (\lambda, x)
        &= A_2 \lambda^2 
        + A_{1} (x) \lambda 
        + A_0 (x).
    \end{align}
\end{itemize}

In the papers \cite{kapaev1999note}, \cite{suleimanov2008quantizations}, and \cite{joshi2009linearization}, it was established that all these linearizations of the \ref{eq:P2sys} system are equivalent. The equivalence between the HTW- and FN-type pairs is given by a gauge transformation supplemented with a change of the spectral parameter. 

\begin{prop}[see \cite{kapaev1999note}] \label{thm:P2_HTW_FN}
Let the functions $W (\mu, x)$ and $Z (\zeta, x)$ be solutions of the linear system \eqref{eq:linsys} with \eqref{eq:HTWpair_form} and \eqref{eq:FNpair_form}, respectively.
Then the {\rm{HTW}}-type pair is gauge equivalent to the {\rm{FN}}-type pair by the Fabri-type map\rm{:}
\begin{align}
    &&
    Z (\zeta, x)
    &= G (\zeta) W (\mu, x),
    &
    G (\zeta)
    &= 
    \begin{pmatrix} 
    \zeta^{-\frac12} & - \zeta^{\frac12} 
    \\ 
    \zeta^{-\frac12} & \zeta^{\frac12} 
    \end{pmatrix}
    ,
    &
    \mu
    &= - \zeta^2.
    &&
\end{align}
\end{prop}
In order to derive a transformation that maps the HTW-type pair to the JM-type pair, one should consider a generalized Laplace transformation \cite{suleimanov2008quantizations}, \cite{joshi2009linearization}.

\begin{prop}[see  \cite{joshi2009linearization}] \label{thm:P2_HTW_JM}
Let the functions $W (\mu, x)$ and $Y (\lambda, x)$ be solutions of the linear problem \eqref{eq:linsys} with \eqref{eq:HTWpair_form} and \eqref{eq:JMpair_form}, respectively.
Then the {\rm{HTW}}-type pair is equivalent to the {\rm{JM}}-type pair by the generalized Laplace transformation\rm{:}
\begin{align}
    W (\mu, x)
    &= \int_L \, 
    e^{\lambda \, \mu} \, 
    Y(\lambda, x) \, d\lambda
    .
\end{align}
\end{prop}

\begin{rem}
\label{rem:countour}
The contour $L$ is chosen to eliminate any remained terms that arise from the integration-by-parts. 
Such a choice can always be made, since the solutions of the \Painleve equations have polynomial asymptotics only. 
\end{rem}

In the matrix case, the \ref{eq:P2sys} system has three non-equivalent generalizations, labeled as the \ref{eq:P20sys}, \ref{eq:P21sys}, and \ref{eq:P22sys} systems, which have constant arbitrary matrix parameters and possess isomonodromic representations \eqref{eq:linsys} with \eqref{eq:HTWpair_form} or \eqref{eq:FNpair_form}  \cite{Adler_Sokolov_2020_1}, i.e. the systems follow from the zero-curvature condition
\begin{align}
    \label{eq:zrc}
    \partial_x A
    - \partial_{\lambda} B
    &= [B, A].
\end{align}
In general, the systems can be written in the form
\begin{align} 
\label{eq:P2matsys}
     &\left\{
    \begin{array}{lcl}
         f'
         &=& - f^2 + g - \frac12 x \, \mathbb{I} - c_1,  
         \\[2mm]
         g'
         &=& 2 g f + \beta [g, f] 
         + c_2 f + f c_3 + c_4,
    \end{array}
    \right.
\end{align}
where $f = f(x)$, $g = g(x)$ are matrix-valued functions, $c_i$, $i = 1, 2, 3, 4$ are arbitrary constant matrices, and $\beta \in \mathbb{C}$. More precisely, in \eqref{eq:P2matsys} one needs to set
\begin{align}
    \text{\ref{eq:P20sys}}:&
    &
    \beta
    &= -1
    ,
    &
    c_1
    &= b
    ,
    &
    c_2
    &= 0
    ,
    &
    c_3
    &= 0
    ,
    &
    c_4
    &= \theta \, \mathbb{I},
    &
    \theta 
    &\in \mathbb{C}
    ;
    \\[1mm]
    \text{\ref{eq:P21sys}}:&
    &
    \beta
    &= -2
    ,
    &
    c_1
    &= 0
    ,
    &
    c_2
    &= 0
    ,
    &
    c_3
    &= 0
    ,
    &
    c_4
    &= a
    ;
    \\[1mm]
    \text{\ref{eq:P22sys}}:&
    &
    \beta
    &= -3
    ,
    &
    c_1
    &= b
    ,
    &
    c_2
    &= 2 b
    ,
    &
    c_3
    &= - 2 b
    ,
    &
    c_4
    &= a,
    &
    [b, a]
    &= 2b
    .
\end{align}

Starting from these systems, instead of the matrix case, we can consider non-abelian analogs for the~\ref{eq:P2sys} system as follows (see also \cite{bobrova2022classification}, \cite{bobrova2022non}, \cite{bobrova2023classification}). Let $\mathcal{R}$ be an associative division ring with the unit $\mathbf{1}$ over the field $\mathbb{C}$. For brevity, we will omit the unit. We suppose that elements $f$, $g$, and $c_i$ are generators of $\mathcal{R}$ and the element $x$ belongs to a center $\mathcal{Z} (\mathcal{R})$ of the ring $\mathcal{R}$. We also suppose that series are allowed in $\mathcal{R}$. 
In addition to a product in $\mathcal{R}$, one is able to introduce an opposite product, the transition to which is defined by the map below. 
\begin{defn}
\label{defn:tau}
An involution $\tau: \mathcal{R} \to \mathcal{R}$ is called \textit{a transposition} if
\begin{enumerate}
    \item
    it acts trivially on the generators of $\mathcal{R}$, e.g. $\tau(f) = f$, and
    \vspace{1mm}
    
    \item 
    for any $F$, $G \in \mathcal{R}$ we have $\tau(F \, G) = \tau(G) \, \tau(F)$.
\end{enumerate}
\end{defn}
For example, $\tau\left([f, g]\right) = - [f, g]$, where $[f, g] := f g - g f$. 

Since we would like to work with systems of non-abelian ODEs, let us equip the ring by a derivation.
\begin{defn}
\label{defn:der}
An $\mathbb{C}$-linear map $\partial_x : \mathcal{R} \to \mathcal{R}$ is called \textit{a derivation of $\mathcal{R}$} if it satisfies the given properties:
\begin{enumerate}
    \item 
    $\partial_x (x) = \mathbf{1}$,
    \vspace{1mm}
    
    \item
    $\partial_x(\alpha) = 0$ for any $\alpha \in \mathbb{C}$,
    \vspace{1mm}
    
    \item
    $\partial_x (F \, G) = \partial_x (F) \, G + F \, \partial_x(G)$ for any $F$, $G \in \mathcal{R}$ (\textit{the Leibniz rule}).
\end{enumerate}
\end{defn}
For simplicity, we will use $\phantom{f}'$ instead of $\partial_x$ and denote $\partial_x(f) = f'$, $\partial_x(f') = f''$, and so on. Since $c_i$ are arbitrary non-abelian constants, $\partial_x(c_i) = 0$. 
Let $\lambda \in \mathcal{Z} (\mathcal{R})$ be \textit{a spectral parameter}. In addition to the properties given in Definition \ref{defn:der}, the derivation $\partial_{\lambda}$ must also satisfy the following properties
\begin{align}
    \partial_{\lambda} (\lambda)
    &= \mathbf{1},
    &
    \partial_{\lambda} (x)
    &= 0,
    &
    \partial_{\lambda} (f)
    &= 0,
    &
    \partial_{\lambda} (g)
    &= 0,
    &
    (\lambda)'
    &= 0.
\end{align}

For the non-abelian $\text{P}_2$ systems of the form \eqref{eq:P2matsys}, different linearizations \eqref{eq:linsys} of the HTW-, FN-, and JM-types are considered. Namely, for each of the systems we present such Lax pairs. 

We extend the list of the HTW-type pairs by those obtained in the paper \cite{Bobrova_Sokolov_2022} by the limiting transitions of the corresponding Lax pairs for the matrix $\rm{P}_4$-type systems. They are given by the \ref{eq:P20_HTWpair} pair for \ref{eq:P20sys}, \ref{eq:P21_HTWpair_2} and \ref{eq:P21_HTWpair} for \ref{eq:P21sys}, and \ref{eq:P22_HTWpair} for the \ref{eq:P22sys} system. To derive the FN-type pairs from them, one is able to use the same Fabry-type map as in Proposition \ref{thm:P2_HTW_FN}. The resulting FN-type pairs are listed in Appendix~\ref{ap:FNpairs}.

In Section \ref{sec:HTWtoJM}, we suggest a method how to construct in the non-abelian case the JM-type pairs from the HTW-pairs by a non-abelian analog of the generalized Laplace transformation given in Proposition \ref{thm:P2_HTW_JM}. Namely, we rewrite a HTW pair in a special form to which it is convenient to apply the transformation. Unlike the paper \cite{joshi2009linearization}, we avoid an auxiliary $3\times3$-pair that is linear in the spectral parameter. Similarly to Remark \ref{rem:countour}, we assume that one is able to eliminate terms which arise from the integration-by-parts. But, actually, that is an extra problem to prove that such a contour can be chosen. We leave this issue for a further research. 

\begin{rem}
    In the paper \cite{Bobrova_Sokolov_2022}, the authors suggested a method for constructing a non-abelian isomonodromic Lax pair for a given \PPainleve-type system. They used an idea of a non-abelianization of the well-known commutative pairs. Namely, we need to construct a non-abelian ansatz for a non-abelian system and determine the undefined coefficients, by using the zero-curvature condition \eqref{eq:zrc}. The main problem of this method is how to construct a non-abelian ansatz whose coefficients would be determined from the zero-curvature representation. 
\end{rem}

\begin{rem}
    There is another method based on the symmetry reductions of integrable PDEs \cite{Adler_Sokolov_2020_1}. Besides the fact that the reductions in the non-abelian case are much more complicated than in the commutative setting, it is not clear which PDE is reduced to a desired system of the \PPainleve-type.
\end{rem}

It turns out that the \ref{eq:P20sys} system has a polynomial JM-type pair, the \ref{eq:P20_JMpair} pair, which can be generalized to a fully non-commutative case (see Remark \ref{rem:fnc_P20}). The \ref{eq:P21sys} system has polynomial and non-polynomial JM-type pairs (pairs \ref{eq:P21_JMpair_2} and \ref{eq:P21_JMpair}, respectively). In the case of the \ref{eq:P22sys} system, the JM-pair is degenerate\footnote{We follow the terminology suggested in \cite{joshi2007linearization}, \cite{joshi2009linearization}.} (see the \ref{eq:P22_JMpair} pair). As a result, we present for each of the non-abelian $\rm{P}_2$ systems linearizations of the HTW, FN, and JM types. Namely, we obtain the following
\begin{thm}
The non-abelian \ref{eq:P20sys}, \ref{eq:P21sys}, and \ref{eq:P22sys} systems have linearizations of the {\rm{HTW}}, {\rm{FN}}, and {\rm{JM}} types.
\end{thm}

Section \ref{sec:monodromy_surf} is devoted to the non-abelian monodromy surfaces associated with the FN and JM type linearizations. In order to derive them, we need to determine the monodromy data that consists of the topological monodromy, the formal monodromy, and the Stokes matrices. This data is defined by a formal solution near a singular point. Using Proposition \ref{thm:irrsol}, that is a non-abelian generalization of Proposition 2.2 in \cite{jimbo1981monodromy1} and a generalization of Proposition 4.1 in \cite{Bertola2018}, one is able\footnote{Under some assumptions.} to find the monodromy data in the non-commutative case. Once the monodromy data are determined\footnote{Note that the monodromy data are integrals of motions.}, we can construct the monodromy relation that should parameterize the space of solutions of the corresponding \PPainleve-type system. The equation following from the monodromy relation is called \textit{the~monodromy~surface}.

Regarding the \ref{eq:P20sys} system, the monodromy surfaces associated with the \ref{eq:P20_FNpair} and \ref{eq:P20_JMpair} pairs can be derived. 
\begin{prop}
Let $x_i$, $i = 1, 2, 3$, and $q$ belong to $\mathcal{R}$, and $\alpha$, $\theta \in \mathbb{C}$. Then the monodromy surfaces related to the \ref{eq:P20_FNpair} and \ref{eq:P20_JMpair} pairs, respectively, are given by the equations
\begin{align}
    x_1 \, x_2 \, x_3 
    + x_1 
    + x_2 
    + x_3
    - 2 \sin(\pi \, \theta) \, q^{-1}
    &= 0,
    \\
    x_1 \, x_2 \, x_3
    - x_1 
    - x_2 \, (\alpha \, q^2)
    - x_3
    + (1 + \alpha) \, q^2
    &= 0.
\end{align}
\end{prop}
We would like to note that $q' = 0$, i.e. this is a conserved quantity for the \ref{eq:P20sys} system. In the commutative setting, $q = 1$ and, thus, the relations above become the well-known affine cubics for the second \Painleve equation. Note also that in the commutative case these equations are equivalent by a simple scaling that cannot be generalized to the non-abelian setting.
Regarding the remaining systems \ref{eq:P21sys} and \ref{eq:P22sys}, the monodromy data are not isomonodromic (it follows from Propositions \ref{thm:FNsol} and \ref{thm:JMsol}) and, thus, we cannot parameterize their solutions by the Stokes multipliers. But, in fact, one can ask about a gauge-transformation that makes the monodromy data isomonodromic. As far as the author knows, such a transformation does not exist.

Inspired by the papers \cite{mazzocco2012confluence}, \cite{chekhov2017painleve}, we devote Section \ref{sec:PoisStr} to a short discussion on a natural Poisson structure connected with the considered non-commutative monodromy surfaces. 
In the commutative setting, the monodromy surfaces can be regarded as a family of affine cubics. These cubic surfaces have a volume form $\Omega$ defined by the Poincar\'e residue formulas. This 2-form is holomorphic on the non-singular part of the affine cubic and determines the Poisson brackets on the surface in the usual way. Generalizing this construction to the non-abelian case, we start from the opposite side, by constructing non-abelian Poisson brackets associated with non-commutative cubics. However, the development of this construction requires further investigation and will be addressed elsewhere. 

\subsection*{Acknowledgments}
The author is thankful to V.~Poberezhnyi, V.~Retakh, V.~Rubtsov,
V.~Sokolov, and B. Suleimanov for lots of invaluable suggestions and comments. The author is also grateful to I. Gaiur for fruitful discussions and to I.H.E.S. for the hospitality and excellent working facilities during the author's visit. This research was partially supported by the International Laboratory of Cluster Geometry HSE, RF Government grant № 075-15-2021-608 and by Young~Russian Mathematics award. 
\newline
The author would also like to thank the referee for the suggested improvements.

\section{The HWT-JM correspondence}
\label{sec:HTWtoJM}
In this section we disscuss a method how to derive a Jimbo-Miva pair by a given pair of the Harnard-Tracy-Widom type.\footnote{This approach is motivated by the problem of construction the Jimbo-Miwa type pair for the \ref{eq:P22sys} system.} Recall that $x$, $\mu$, $\lambda \in \mathcal{Z} (\mathcal{R})$, where $\mu$ and $\lambda$ are spectral parameters, and matrices below are defined over $\mathcal{R}$.

The scheme of construction is the following. 
\begin{itemize}
\item 
Consider the linear problem of the HTW-type:
\begin{align}
    \left\{
    \begin{array}{lcl}
         \partial_{\mu} W (\mu, x)
         &=& \brackets{
         A_1 \mu + A_0 + A_{-1} \mu^{-1}
         } 
         \, W (\mu, x)
         ,  
         \\[3mm]
         \partial_{x} W (\mu, x)
         &=& \brackets{
         B_1 \mu + B_0
         }
         \, W (\mu, x)
         .
    \end{array}
    \right.
\end{align}

\item
Suppose that the first equation can be rewritten in the form
\begin{align}
    \label{eq:HTWpair_altform}
    \brackets{
    C_1 \mu + C_2
    } \, \partial_{\mu} W (\mu, x)
    &= \brackets{
    C_3 \mu + C_4
    } \, W (\mu, x).
\end{align}
So, we should solve the following matrix equation
\begin{align}
    \label{eq:Ccond}
    \brackets{
    C_3 \mu + C_4
    }
    - \brackets{
    C_1 \mu + C_2
    }
    \brackets{
    A_1 \mu + A_0 + A_{-1} \mu^{-1}
    } 
    &= 0.
\end{align}

\item
If the equation has a solution, then it is easy to find a linear problem for the function $Y (\lambda, x)$, that is related to $W (\mu, x)$ by the generalized Laplace transform:
\begin{align}
    \label{eq:genLaplacetransform}
    W (\mu, x)
    &= \int_{L} \, e^{\lambda \, \mu} \, Y (\lambda, x) \, d \lambda.
\end{align}
Namely, the linear problem for $Y (\lambda, x)$ reads as
\begin{align}
    \label{eq:JMsys}
    \left\{
    \begin{array}{rcl}
         \brackets{
         C_1 \lambda - C_3
         } \, 
         \partial_{\lambda} Y (\lambda, x)
         &=& \brackets{
         C_2 \lambda - C_1 - C_4
         } 
         \, Y (\lambda, x)
         ,  
         \\[3mm]
         \partial_{x} Y (\lambda, x)
         &=& - B_1 \, \partial_{\lambda} Y (\lambda, x)
         + B_0 \, Y (\lambda, x)
         .
    \end{array}
    \right.
\end{align}

\item
Set $C (\lambda) = C_1 \lambda - C_3$. If there exists a matrix $C^{-1}$ such that
\begin{align}
    \label{eq:invmatrix_cond}
    C \, C^{-1}
    &= C^{-1} \, C
    = \mathbb{I},
\end{align}
then the linear problem for $Y (\lambda, x)$ has the JM-type. Otherwise, according to the terminology in \cite{joshi2007linearization}, \cite{joshi2009linearization}, the function $Y (\lambda, x)$ admits \textit{a degenerate pair}.
\end{itemize}

\begin{defn}
\label{def:invmat}
If  for a generic matrix $C$ condition \eqref{eq:invmatrix_cond} holds, we will call such a matrix \textit{invertible}.
\end{defn}

Note that matrices $C_k$ are defined up to an invertible factor $C \in \Mat_n (\mathcal{R})$, namely, matrices $\tilde C_k = C \, C_{k}$ are also solutions of \eqref{eq:Ccond}. Due to this obvious observation, we have the following 
\begin{lem}
\label{thm:uniqC}
Let $C$ be an invertible generic matrix. Then a solution of \eqref{eq:Ccond} is defined up to a left factor $C$ and does not depend on it.
\end{lem}

Following the scheme, we present for each of the non-abelian $\rm{P}_2$ systems a linear problem of the JM-type. In the cases of the \ref{eq:P20sys} and \ref{eq:P21sys} systems, these pairs are polynomial in $f$ and $g$. The second pair of the HTW-type for the \ref{eq:P21sys} system transforms into a non-polynomial pair of the JM-type. In the case of \ref{eq:P22sys}, the JM pair is degenerate.

\subsection{Case \texorpdfstring{\ref{eq:P20sys}}{P20}}
Consider the system
\begin{align}
    \label{eq:P20sys}   
    \tag*{$\rm{P}_2^0$}
    \left\{
    \begin{array}{lcl}
         f'
         &=& - f^2 + g - \frac12 x - b,  
         \\[2mm]
         g'
         &=& g f + f g + \theta,
    \end{array}
    \right.
\end{align}
where $f$, $g$, $b \in \mathcal{R}$, $x \in \mathcal{Z} (\mathcal{R})$, $\theta \in \mathbb{C}$, and $b' = 0$. This system possesses the HTW pair (see \cite{Irfan_2012}, \cite{Adler_Sokolov_2020_1}, \cite{Bobrova_Sokolov_2022})
\begin{align} \label{eq:P20_HTWpair}
    \tag*{$\rm{HTW}_2^0$}
    \left\{
    \begin{array}{lcl}
        \partial_{\mu} W (\mu, x)
        &=&
        \LieBrackets{
        \begin{pmatrix}
        0
        & 
        2
        \\[0.9mm]
        0 
        & 
        0
        \end{pmatrix}
        \mu 
        + 
        \begin{pmatrix}
        - 2 f
        & 
        2 f^2 - g + x + 2 b
        \\[0.9mm]
        - 2
        & 
        2 f
        \end{pmatrix}
        + 
        \begin{pmatrix}
        0 & 0 
        \\[0.9mm]
        - g & - \theta
        \end{pmatrix}
        \mu^{-1}
         } W (\mu, x),
        \\[5mm]
        \partial_{x} W (\mu, x)
        &=& \LieBrackets{
        \begin{pmatrix}
        0
        & 
        1
        \\[0.9mm]
        0 
        & 
        0
        \end{pmatrix}
        \mu 
        + 
        \begin{pmatrix}
        - f
        & 
        0 
        \\[0.9mm]
        - 1
        & 
        f
        \end{pmatrix}
         } W (\mu, x).
    \end{array}
    \right.
\end{align}

Following the procedure mentioned above, we obtain
\begin{prop}
The \ref{eq:P20_HTWpair} pair transforms by the map \eqref{eq:genLaplacetransform} to the following pair of the {\rm{JM}}-type{\rm:}
\begin{align} \label{eq:P20_JMpair}
    \tag*{$\rm{JM}_2^0$}
    \left\{
    \begin{array}{lcl}
        \partial_{\lambda} Y (\lambda, x)
        &=&
        \left[
        \begin{pmatrix}
        2 & 0
        \\[0.9mm]
        0 & 0
        \end{pmatrix}
        \lambda^2
        + 
        \begin{pmatrix}
        0 & - 2 f^2 + g - x - 2 b
        \\[0.9mm]
        - 2 & 0
        \end{pmatrix}
        \lambda
        \right.
        \\[3mm]
        && \qquad
        \left.
        +
        \begin{pmatrix}
        - 2 f^2 + g
        & 
        2 f^3 - f g + f (x + 2 b) + \theta - 1
        \\[0.9mm]
        - 2 f 
        & 
        2 f^2 - g + x + 2 b
        \end{pmatrix}
         \right] Y (\lambda, x),
        \\[5mm]
        \partial_{x} Y (\lambda, x)
        &=& \LieBrackets{
        \begin{pmatrix}
        1 & 0 
        \\[0.9mm] 
        0 & 0
        \end{pmatrix} 
        \lambda
        +
        \begin{pmatrix}
        0 & - f^2 + \frac12 g - \frac12 x - b
        \\[0.9mm] 
        - 1 & f
        \end{pmatrix}
         } Y (\lambda, x).
    \end{array}
    \right.
\end{align}
\end{prop}
\begin{rem}
A gauge-equivalent Lax pair was presented in \cite{Adler_Sokolov_2020_1}. The authors did not study the relation between these two different pairs for the \ref{eq:P20sys} system.
\end{rem}
\begin{proof}

Solving \eqref{eq:Ccond}, we obtain the following matrices $C_k = (c_{k, ij})_{i, j = 1, 2}$:
\begin{gather}
    \begin{aligned}
        C_1
        &= 
        \begin{pmatrix}
        0 & c_{1, 12}
        \\[0.9mm]
        0 & c_{1, 22}
        \end{pmatrix}
        ,
        &&&
        C_2
        &= 
        \begin{pmatrix}
        c_{2, 11} & 0 
        \\[0.9mm]
        c_{2, 21} & 0
        \end{pmatrix}
        ,
        &&&
        C_3
        &= 
        \begin{pmatrix}
        - 2 c_{1, 12}
        & 
        2 c_{1, 12} f + 2 c_{2, 11}
        \\[0.9mm]
        - 2 c_{1, 22}
        & 
        2 c_{1, 22} f + 2 c_{2, 21}
        \end{pmatrix}
        ,
    \end{aligned}
    \\[3mm]
    C_{4}
    = 
    \begin{pmatrix}
    - c_{1, 12} g - 2 c_{2, 11} f
    &
    - \theta c_{1, 12}
    + c_{2, 11} (2 f^2 - g + x + 2 b)
    \\[0.9mm]
    - c_{1, 22} g - 2 c_{2, 21} f
    &
    - \theta c_{1, 22}
    + c_{2, 21} (2 f^2 - g + x + 2 b)
    \end{pmatrix}.
\end{gather}
Note that they can be rewritten as
\begin{gather}
    \begin{aligned}
        C_1
        &= 
        \begin{pmatrix}
        c_{1, 12} & c_{2, 11}
        \\[0.9mm]
        c_{1,22} & c_{2, 21}
        \end{pmatrix}
        \begin{pmatrix}
        0 & 1
        \\[0.9mm]
        0 & 0
        \end{pmatrix}
        ,
        &&&
        C_2
        &= 
        \begin{pmatrix}
        c_{1, 12} & c_{2, 11}
        \\[0.9mm]
        c_{1,22} & c_{2, 21}
        \end{pmatrix}
        \begin{pmatrix}
        0 & 0
        \\[0.9mm]
        1 & 0
        \end{pmatrix}
        ,
    \end{aligned}
    \\[3mm]
    \begin{aligned}
        C_3
        &= 
        \begin{pmatrix}
        c_{1, 12} & c_{2, 11}
        \\[0.9mm]
        c_{1,22} & c_{2, 21}
        \end{pmatrix}
        \begin{pmatrix}
        -2 & 2 f
        \\[0.9mm]
        0 & 2
        \end{pmatrix}
        ,
        &&&
        C_{4}
        &= 
        \begin{pmatrix}
        c_{1, 12} & c_{2, 11}
        \\[0.9mm]
        c_{1,22} & c_{2, 21}
        \end{pmatrix}
        \begin{pmatrix}
        - g & - \theta 
        \\[0.9mm]
        - 2 f & 2 f^2 - g + x + 2 b
        \end{pmatrix}
        ,
    \end{aligned}
\end{gather}
i.e. a solution of equation \eqref{eq:Ccond} is unique up to a factor of the form 
$ \small
        \begin{pmatrix}
        c_{1, 12} & c_{2, 11}
        \\
        c_{1,22} & c_{2, 21}
        \end{pmatrix}
$ (see Proposition \ref{thm:uniqC}). Then the first equation of system \eqref{eq:JMsys} reads
\begin{align} 
    \label{eq:eq1JMP20}
    \begin{pmatrix}
    2 & \lambda - 2 f 
    \\[0.9mm]
    0 & - 2
    \end{pmatrix}
    \, 
    \partial_{\lambda} Y (\lambda, x)
    &=
    \begin{pmatrix}
    g & \theta - 1
    \\[0.9mm]
    \lambda + 2 f & - 2 f^2 + g - x - 2 b
    \end{pmatrix}
    \, Y (\lambda, x).
\end{align}
Using condition \eqref{eq:invmatrix_cond}, one is able to find an inverse matrix for the matrix $C (\lambda) = \small \begin{pmatrix} 2 & \lambda - 2 f \\ 0 & - 2 \end{pmatrix}$:
\begin{gather}
    \label{eq:invGP20}
    C^{-1}
    = \tfrac14
    \begin{pmatrix}
    2 & \lambda - 2 f
    \\[0.9mm]
    0 & - 2
    \end{pmatrix}
    .
\end{gather}
Multiplying \eqref{eq:eq1JMP20} on the left by $C^{-1}$, the system \eqref{eq:JMsys} turns into the \ref{eq:P20_JMpair} pair, where the scaling $\lambda \mapsto \tfrac12 \lambda$ was made, for the \ref{eq:P20sys} system.
\end{proof}

\begin{rem}
\label{rem:fnc_P20}
As it was suggested in the paper \cite{Bobrova_Sokolov_2022}, there is a connection between the commutative and non-commutative ``independent'' variables. Introducing the non-commutative variable $\bar{x} \in \mathcal{R}$ by the formula $\bar x = x + 2 b$, one is able to obtain the system
\begin{align} 
    \label{eq:P2ncsys}
    &&
    &\left\{
    \begin{array}{lcl}
         f'
         &=& - f^2 + g - \tfrac12 \bar x, 
         \\[2mm]
         g'
         &=& f g + g f + \theta,
    \end{array}
    \right.
    &
    f, \,\,
    g, \,\,
    \bar x
    &\in \mathcal{R},
    &
    \theta
    &\in \mathbb{C},
    &&
\end{align}
that is equivalent to the equation
\begin{align} \label{eq:P2nc}
    f''
    &= 2 f^3 
    + \tfrac12 \bar x \, f 
    + \tfrac12 f \, \bar x 
    + \brackets{\theta - \tfrac12}.
\end{align}
This equation was first considered in the paper \cite{Retakh_Rubtsov_2010}, where the authors have assumed that the equation $\partial_{\bar x} (\bar x) = 1$ has a solution in $\mathcal{R}$. The system \eqref{eq:P2ncsys} has the following isomomodromic Lax pair
\begin{gather} \label{eq:P2nc_JMpair}
\begin{aligned}
    A (\lambda, \bar x)
    &= 
    \begin{pmatrix}
    2 & 0
    \\[0.9mm]
    0 & 0
    \end{pmatrix}
    \lambda^2
    + 
    \begin{pmatrix}
    0 & - 2 f^2 + g - \bar x
    \\[0.9mm]
    - 2 & 0
    \end{pmatrix}
    \lambda
    +
    \begin{pmatrix}
    - 2 f^2 + g
    & 
    2 f^3 - f g + f \bar x + \theta - 1
    \\[0.9mm]
    - 2 f 
    & 
    2 f^2 - g + \bar x
    \end{pmatrix}
    ,
    \\[2mm]
    B (\lambda, \bar x)
    &= 
    \begin{pmatrix}
    1 & 0 
    \\[0.9mm] 
    0 & 0
    \end{pmatrix} 
    \lambda
    +
    \begin{pmatrix}
    0 & - f^2 + \frac12 g - \frac12 \bar x
    \\[0.9mm] 
    - 1 & f
    \end{pmatrix}
    ,
\end{aligned}
\end{gather}
that can be obtained from \ref{eq:P20_JMpair} just by the change of the ``independent'' variable.
\end{rem}

\subsection{Case \texorpdfstring{\ref{eq:P21sys}}{P21}}
Let $f$, $g$, and $a$ belong to $\mathcal{R}$, $x \in \mathcal{Z}(\mathcal{R})$, and $a' = 0$. One can verify that the system
\begin{align}
    \label{eq:P21sys}   
    \tag*{$\rm{P}_2^1$}
    \left\{
    \begin{array}{lcl}
         f'
         &=& - f^2 + g - \frac12 x,  
         \\[2mm]
         g'
         &=& 2 f g + a
    \end{array}
    \right.
\end{align}
is equivalent to the compatibility condition of the following linear system of the HTW-type \cite{Adler_Sokolov_2020_1}, \cite{Bobrova_Sokolov_2022}:
\begin{align} \label{eq:P21_HTWpair_2}
    \tag*{$\rm{HTW}_2^1$}
    \left\{
    \begin{array}{lcl}
        \partial_{\mu} W (\mu, x)
        &=&
        \LieBrackets{
        \begin{pmatrix}
        0
        & 
        2
        \\[0.9mm]
        0 
        & 
        0
        \end{pmatrix}
        \mu 
        + 
        \begin{pmatrix}
        - 2 f
        & 
        2 f^2 - g + x
        \\[0.9mm]
        - 2
        & 
        2 f
        \end{pmatrix}
        + 
        \begin{pmatrix}
        0 & 0 
        \\[0.9mm]
        - g & - [f, g] - a
        \end{pmatrix}
        \mu^{-1}
         } W (\mu, x),
        \\[5mm]
        \partial_{x} W (\mu, x)
        &=& \LieBrackets{
        \begin{pmatrix}
        0
        & 
        1
        \\[0.9mm]
        0 
        & 
        0
        \end{pmatrix}
        \mu 
        + 
        \begin{pmatrix}
        - f
        & 
        0 
        \\
        - 1
        & 
        f
        \end{pmatrix}
         } W (\mu, x).
    \end{array}
    \right.
\end{align}
Using this pair, we obtain a pair of the JM-type. 
\begin{prop}
The \ref{eq:P21_HTWpair_2} pair transforms by the map \eqref{eq:genLaplacetransform} to the following pair of the {\rm{JM}}-type{\rm:}
\begin{align} \label{eq:P21_JMpair_2}
    \tag*{$\rm{JM}_2^1$}
    \left\{
    \begin{array}{lcl}
        \partial_{\lambda} Y (\lambda, x)
        &=&
        \left[
        \begin{pmatrix}
        2 & 0
        \\[0.9mm]
        0 & 0
        \end{pmatrix}
        \lambda^2
        + 
        \begin{pmatrix}
        0 & - 2 f^2 + g - x
        \\[0.9mm]
        -2 & 0
        \end{pmatrix}
        \lambda
        \right.
        \\[3mm]
        && \qquad
        \left.
        +
        \begin{pmatrix}
        - 2 f^2 + g 
        & 
        2 f^3 - g f + x f + a - 1
        \\[0.9mm]
        - 2 f
        & 
        2 f^2 - g + x
        \end{pmatrix}
         \right] Y (\lambda, x),
        \\[5mm]
        \partial_{x} Y (\lambda, x)
        &=& \LieBrackets{
        \begin{pmatrix}
        1 & 0 
        \\[0.9mm] 
        0 & 0
        \end{pmatrix} 
        \lambda
        +
        \begin{pmatrix}
        0 & - f^2 + \frac12 g - \frac12 x
        \\[0.9mm] 
        - 1 & f
        \end{pmatrix}
         } Y (\lambda, x).
    \end{array}
    \right.
\end{align}
\end{prop}
\begin{proof}
The condition \eqref{eq:Ccond} leads to the following matrices $C_k = (c_{k, ij})_{i, j = 1, 2}$:
\begin{gather}
    \begin{aligned}
        C_1
        &= 
        \begin{pmatrix}
        c_{1, 12} & c_{2, 11}
        \\[0.9mm]
        c_{1,22} & c_{2, 21}
        \end{pmatrix}
        \begin{pmatrix}
        0 & 1 \\[0.9mm] 0 & 0
        \end{pmatrix}
        ,
        &&&
        C_2
        &= 
        \begin{pmatrix}
        c_{1, 12} & c_{2, 11}
        \\[0.9mm]
        c_{1,22} & c_{2, 21}
        \end{pmatrix}
        \begin{pmatrix}
        0 & 0 \\[0.9mm] 1 & 0
        \end{pmatrix}
        ,
    \end{aligned}
    \\[3mm]
    \begin{aligned}
        C_3
        &= 
        \begin{pmatrix}
        c_{1, 12} & c_{2, 11}
        \\[0.9mm]
        c_{1,22} & c_{2, 21}
        \end{pmatrix}
        \begin{pmatrix}
        - 2 & 2 f \\[0.9mm] 0 & 2
        \end{pmatrix}
        ,
        &&&
        C_{4}
        &= 
        \begin{pmatrix}
        c_{1, 12} & c_{2, 11}
        \\[0.9mm]
        c_{1,22} & c_{2, 21}
        \end{pmatrix}
        \begin{pmatrix}
        - g & - [f, g] - a \\[0.9mm]
        - 2 f & 2 f^2 - g + x
        \end{pmatrix}
        .
    \end{aligned}
\end{gather}
Therefore, the first equation of \eqref{eq:JMsys} is
\begin{align} 
    \label{eq:eq1JMP21_2}
    \begin{pmatrix}
    2 & \lambda - 2 f 
    \\[0.9mm]
    0 & - 2
    \end{pmatrix}
    \, 
    \partial_{\lambda} Y (\lambda, x)
    &=
    \begin{pmatrix}
    g & [f, g] + a - 1
    \\[0.9mm]
    \lambda + 2 f & - 2 f^2 + g - x
    \end{pmatrix}
    \, Y (\lambda, x).
\end{align}
Multiplying \eqref{eq:eq1JMP21_2} on the left by \eqref{eq:invGP20}, we get the \ref{eq:P21_JMpair_2} pair (after the scaling $\lambda \mapsto \tfrac12 \lambda$) for the \ref{eq:P21sys} system.
\end{proof}

Another pair for \ref{eq:P21sys} is given by \cite{Bobrova_Sokolov_2022}
\begin{align} \label{eq:P21_HTWpair}
    \tag*{$\rm{HTW'}_2^1$}
    \left\{
    \begin{array}{lcl}
        \partial_{\mu} W (\mu, x)
        &=&
        \LieBrackets{
        \begin{pmatrix}
        0
        & 
        2
        \\[0.9mm]
        0 
        & 
        0
        \end{pmatrix}
        \mu 
        + 
        \begin{pmatrix}
        - 2 f
        & 
        2 f^2 - g + x
        \\[0.9mm]
        - 2
        & 
        2 f
        \end{pmatrix}
        + 
        \begin{pmatrix}
        a & 0 
        \\[0.9mm]
        - g & 0
        \end{pmatrix}
        \mu^{-1}
         } W (\mu, x),
        \\[5mm]
        \partial_{x} W (\mu, x)
        &=& \LieBrackets{
        \begin{pmatrix}
        0
        & 
        1
        \\[0.9mm]
        0 
        & 
        0
        \end{pmatrix}
        \mu 
        + 
        \begin{pmatrix}
        0
        & 
        0 
        \\
        - 1
        & 
        2 f
        \end{pmatrix}
         } W (\mu, x).
    \end{array}
    \right.
\end{align}
This pair tranfroms to a non-polynomial in $f$ and $g$ pair of the JM-type.
\begin{prop}
The \ref{eq:P21_HTWpair} pair transforms by the map \eqref{eq:genLaplacetransform} to the following pair of the {\rm{JM}}-type{\rm:}
\begin{align} \label{eq:P21_JMpair}
    \tag*{$\rm{JM'}_2^1$}
    \left\{
    \begin{array}{lcl}
        \partial_{\lambda} Y (\lambda, x)
        &=&
        \left[
        \begin{pmatrix}
        2 & 2 a g^{-1}
        \\[0.9mm]
        0 & 0
        \end{pmatrix}
        \lambda^2
        + 
        \begin{pmatrix}
        2 a g^{-1} 
        & 
        - 2 f^2 + g - x 
        - 2 a g^{-1} f - 2 f a g^{-1}
        \\[0.9mm]
        - 2 & - 2 a g^{-1}
        \end{pmatrix}
        \lambda
        \right.
        \\[3mm]
        && \qquad 
        \left.
        +
        \begin{pmatrix}
        - 2 f^2 + g - 2 f a g^{-1}
        & 
        2 f^3 - f g + x f - 1 + 2 f a g^{-1} f
        \\[0.9mm]
        - 2 f - 2 a g^{-1}
        & 
        2 f^2 - g + x + 2 a g^{-1} f
        \end{pmatrix}
         \right] Y (\lambda, x),
        \\[5mm]
        \partial_{x} Y (\lambda, x)
        &=& \LieBrackets{
        \begin{pmatrix}
        1 & a g^{-1}
        \\[0.9mm] 
        0 & 0
        \end{pmatrix} 
        \lambda
        +
        \begin{pmatrix}
        f + a g^{-1} & - f^2 + \frac12 g - \frac12 x - a g^{-1} f
        \\[0.9mm] 
        -1 & 2 f
        \end{pmatrix}
         } Y (\lambda, x).
    \end{array}
    \right.
\end{align}
\end{prop}
\begin{proof}
The given matrices $C_k = (c_{k, ij})_{i, j = 1, 2}$ solve \eqref{eq:Ccond}:
\begin{gather}
    \begin{aligned}
        C_1
        &= 
        \begin{pmatrix}
        c_{1, 12} & c_{2, 11}
        \\[0.9mm]
        c_{1,22} & c_{2, 21}
        \end{pmatrix}
        \begin{pmatrix}
        0 & 1 \\[0.9mm] 0 & 0
        \end{pmatrix}
        ,
        &&&
        C_2
        &= 
        \begin{pmatrix}
        c_{1, 12} & c_{2, 11}
        \\[0.9mm]
        c_{1,22} & c_{2, 21}
        \end{pmatrix}
        \begin{pmatrix}
        0 & 0 \\[0.9mm] 1 & a g^{-1}
        \end{pmatrix}
        ,
    \end{aligned}
    \\[3mm]
    \begin{aligned}
        C_3
        &= 
        \begin{pmatrix}
        c_{1, 12} & c_{2, 11}
        \\[0.9mm]
        c_{1,22} & c_{2, 21}
        \end{pmatrix}
        \begin{pmatrix}
        - 2 & 2 f \\[0.9mm] 0 & 2
        \end{pmatrix}
        ,
        &&&
        C_{4}
        &= 
        \begin{pmatrix}
        c_{1, 12} & c_{2, 11}
        \\[0.9mm]
        c_{1,22} & c_{2, 21}
        \end{pmatrix}
        \begin{pmatrix}
        - g & 0 \\[0.9mm]
        - 2 a g^{-1} - 2 f 
        & 
        2 a g^{-1} f 
        + 2 f^2 - g + x
        \end{pmatrix}
        .
    \end{aligned}
\end{gather}
Then the first equation of \eqref{eq:JMsys} reads as
\begin{align} 
    \label{eq:eq1JMP21}
    \begin{pmatrix}
    2 & \lambda - 2 f 
    \\[0.9mm]
    0 & - 2
    \end{pmatrix}
    \, 
    \partial_{\lambda} Y (\lambda, x)
    &=
    \begin{pmatrix}
    g & -1
    \\[0.9mm]
    \lambda + 2 f + 2 a g^{-1}
    & 
    a g^{-1} \brackets{
    \lambda - 2 f
    }
    - 2 f^2 + g - x
    \end{pmatrix}
    \, Y (\lambda, x).
\end{align}
Multiplying \eqref{eq:eq1JMP21} on the left by \eqref{eq:invGP20} and making the scaling $\lambda \mapsto \tfrac12 \lambda$, we obtain the \ref{eq:P21_JMpair} pair for the \ref{eq:P21sys} system.
\end{proof}

\subsection{Case \texorpdfstring{\ref{eq:P22sys}}{P22}}
Let $f$, $g$, $a$, and $b$ be elements of $\mathcal{R}$, $x \in \mathcal{Z}(\mathcal{R})$, and $a' = b' = 0$. The following system
\begin{align}
    \label{eq:P22sys}   
    \tag*{$\rm{P}_2^2$}
    &&
    &\left\{
    \begin{array}{lcl}
         f'
         &=& - f^2 + g - \frac12 x - b,  
         \\[2mm]
         g'
         &=& 3 f g - g f + 2 b f - 2 f b + a,
    \end{array}
    \right.
    &
    [b, a]
    &= 2 b.
    &&
\end{align}
has the HTW pair of the form \cite{Adler_Sokolov_2020_1}, \cite{Bobrova_Sokolov_2022}
\begin{align} \label{eq:P22_HTWpair}
    \tag*{$\rm{HTW}_2^2$}
    \left\{
    \begin{array}{lcl}
        \partial_{\mu} W (\mu, x)
        &=&
        \left[
        \begin{pmatrix}
        0
        & 
        2
        \\[0.9mm]
        0 
        & 
        0
        \end{pmatrix}
        \mu 
        + 
        \begin{pmatrix}
        - 2 f
        & 
        2 f^2 - g + x + \frac23 b
        \\[0.9mm]
        - 2
        & 
        2 f
        \end{pmatrix}
        \right.
        \\[3mm]
        && \qquad 
        \left.
        + 
        \begin{pmatrix}
        \frac23 b f + \frac12 a 
        & 
        \frac13 b g 
        - \frac23 b f^2
        - \frac13 x b 
        - \frac49 b^2
        \\[0.9mm]
        - g + \frac23 b
        & 
        - [f, g]
        - \frac43 b f + \frac23 f b - \frac12 a
        \end{pmatrix}
        \mu^{-1}
         \right] W (\mu, x),
        \\[5mm]
        \partial_{x} W (\mu, x)
        &=& \LieBrackets{
        \begin{pmatrix}
        0
        & 
        1
        \\[0.9mm]
        0 
        & 
        0
        \end{pmatrix}
        \mu 
        + 
        \begin{pmatrix}
        0
        & 
        - \frac13 b
        \\[0.9mm]
        - 1
        & 
        2 f
        \end{pmatrix}
         } W (\mu, x).
    \end{array}
    \right.
\end{align}

It turns out that this pair maps to a degenerate pair of the JM type. Namely, we have
\begin{prop}
The \ref{eq:P21_HTWpair} pair transforms by the map \eqref{eq:genLaplacetransform} to the degenerate pair of the {\rm{JM}}-type{\rm:}
\begin{align} \label{eq:P22_JMpair}
    \tag*{$\rm{JM}_2^2$}
    \left\{
    \begin{array}{rcl}
        \begin{pmatrix}
        2 & \lambda - 2 f
        \\[0.9mm]
        0 & 0 
        \end{pmatrix}
        \partial_{\lambda} Y (\lambda, x)
        &=&
        \begin{pmatrix}
        g - \frac23 b
        &
        [f, g] + \frac43 b f - \frac23 f b + \frac12 a - 1
        \\[0.9mm]
        0 & 0 
        \end{pmatrix}
        Y (\lambda, x),
        \\[5mm]
        \partial_{x} Y (\lambda, x)
        &=& 
        \begin{pmatrix}
        0 & -1 \\[0.9mm]
        0 & 0
        \end{pmatrix}
        \partial_{\lambda} Y (\lambda, x)
        + 
        \begin{pmatrix}
        0 & - \frac13 b \\[0.9mm]
        - 1 & 2 f
        \end{pmatrix}
        Y (\lambda, x).
    \end{array}
    \right.
\end{align}
\end{prop}
\begin{rem}
For $a \in \mathbb{C}$\footnote{It implies $b \equiv 0$.}, one can present a non-degenerate polynomial Lax pair of the JM-type. 
\end{rem}
\begin{proof}
Solving \eqref{eq:Ccond}, we obtain the following matrices $C_k = (c_{k, ij})_{i, j = 1, 2}$:
\begin{gather}
    \begin{aligned}
        C_1
        &= 
        \begin{pmatrix}
        0 & c_{1, 12}
        \\[0.9mm]
        0 & c_{1, 22}
        \end{pmatrix}
        ,
        &&&
        C_2
        &= 
        \begin{pmatrix}
        0 & 0 
        \\[0.9mm]
        0 & 0
        \end{pmatrix}
        ,
        &&&
        C_3
        &= 
        \begin{pmatrix}
        - 2 c_{1, 12}
        & 
        2 c_{1, 12} f
        \\[0.9mm]
        - 2 c_{1, 22}
        & 
        2 c_{1, 22} f
        \end{pmatrix}
        ,
    \end{aligned}
    \\[3mm]
    C_{4}
    = 
    \begin{pmatrix}
     c_{1, 12} \left(\tfrac23 b - g\right)
    &
    c_{1, 12} \brackets{
    - [f, g] - \tfrac43 b f + \tfrac23 f b
    - \tfrac12 a
    }
    \\[0.9mm]
    c_{1, 22} \left(\tfrac23 b - g\right) 
    &
    c_{1, 22} \brackets{
    - [f, g] - \tfrac43 b f + \tfrac23 f b
    - \tfrac12 a
    }
    \end{pmatrix}.
\end{gather}
They can be rewritten as
\begin{gather}
    \begin{aligned}
        C_1
        &= 
        \begin{pmatrix}
        c_{1, 12} & c_{2, 11}
        \\[0.9mm]
        c_{1,22} & c_{2, 21}
        \end{pmatrix}
        \begin{pmatrix}
        0 & 1
        \\[0.9mm]
        0 & 0
        \end{pmatrix}
        ,
        &&&
        C_2
        &= 
        \begin{pmatrix}
        c_{1, 12} & c_{2, 11}
        \\[0.9mm]
        c_{1,22} & c_{2, 21}
        \end{pmatrix}
        \begin{pmatrix}
        0 & 0
        \\[0.9mm]
        0 & 0
        \end{pmatrix}
        ,
    \end{aligned}
    \\[3mm]
    \begin{aligned}
        C_3
        &= 
        \begin{pmatrix}
        c_{1, 12} & c_{2, 11}
        \\[0.9mm]
        c_{1,22} & c_{2, 21}
        \end{pmatrix}
        \begin{pmatrix}
        -2 & 2 f
        \\[0.9mm]
        0 & 0
        \end{pmatrix}
        ,
        &&&
        C_{4}
        &= 
        \begin{pmatrix}
        c_{1, 12} & c_{2, 11}
        \\[0.9mm]
        c_{1,22} & c_{2, 21}
        \end{pmatrix}
        \begin{pmatrix}
        \tfrac23 b - g 
        & 
        - [f, g] - \tfrac43 b f + \tfrac23 f b
        - \tfrac12 a
        \\[0.9mm]
        0 & 0
        \end{pmatrix}
        ,
    \end{aligned}
\end{gather}
and, thus, the first equation of system \eqref{eq:JMsys} reads
\begin{align} 
    \label{eq:eq1JMP22}
    \begin{pmatrix}
    2 & \lambda - 2 f 
    \\[0.9mm]
    0 & 0
    \end{pmatrix}
    \, 
    \partial_{\lambda} Y (\lambda, x)
    &=
    \begin{pmatrix}
    g - \tfrac23 b 
    & 
    [f, g] + \tfrac43 b f - \tfrac23 f b
    + \tfrac12 a - 1
    \\[0.9mm]
    0 & 0
    \end{pmatrix}
    \, Y (\lambda, x).
\end{align}
Since the matrix $C (\lambda) = \small \begin{pmatrix} 2 & \lambda - 2 f \\ 0 & 0 \end{pmatrix}$ is not invertible, we get the degenerate \ref{eq:P22_JMpair} pair for \ref{eq:P22sys}.
\end{proof}

\section{Non-abelian monodromy surfaces}
\label{sec:monodromy_surf}

In this section we are going to derive non-abelian generalizations of the well-known monodromy surfaces related to the FN and JM pairs. In order to construct them, we need to determine the monodromy data and verify that it has the isomonodromic property, i.e. it should define the set of first integrals\footnote{A definition of a non-abelian first integral for a given non-abelian ODE can be generalized in the natural way (e.g. \cite{mikhailov2000integrable}).}. 

In Subsection \ref{sec:formsol} we obtain a form of formal solutions near a singular point that is defined uniquely by a linearization. If the monodromy data is isomonodromic, then we are able to provide a non-abelian monodromy surface (see Subsection \ref{sec:monsufr}). 

\subsection{Formal solutions}
\label{sec:formsol}

To derive the non-abelian monodromy surfaces, we need to present a formal solution near a singular point. We use the following generalization of Proposition 2.2~in~\cite{jimbo1981monodromy1} to the non-commutative case, whose particular case was discussed in the paper \cite{Bertola2018}.
\begin{prop}
\label{thm:irrsol}
Set $r \geq 0$\footnote{$r$ is called \textit{the Poincar\'e rank} of an irregular singular point. When $r = 0$, the singular point is Fuchsian.} and $\lambda \in \mathcal{Z} (\mathcal{R})$. Let us consider $n \times n$-matrices $A(\lambda)$, $F(\lambda)$, $D(\lambda)$, $T(\lambda)$ of the~form
\begin{gather}
    \label{eq:irrsys_A}
    \begin{aligned}
    A (\lambda)
    &= \sum_{k \geq - r}^{\phantom{k}} A_{k} \lambda^{- k - 1}
    ,
    &&&
    A_{-r}
    &= diag (\alpha_1, \dots, \alpha_n),
    &
    \alpha_i 
    &\neq \alpha_j,
    \qquad
    i \neq j,
    \end{aligned}
    \\
    \label{eq:irrseries}
    \begin{aligned}
    F (\lambda)
    &= \mathbb{I}
    + \sum_{k \geq 1} F_k \lambda^{-k},
    &&&
    D (\lambda)
    &= \mathbb{I}
    + \sum_{k \geq 1} D_k \lambda^{-k},
    &&&
    \partial_{\lambda} T(\lambda)
    &= \sum_{k = - r}^0 T_k \lambda^{- k - 1},
    \end{aligned}
\end{gather}
where 
\begin{itemize}
    \item[(a)] $A(\lambda) \in \Mat_n (\mathcal{R})$ and $A_{-r} \in \Mat_n (\mathcal{Z}(\mathcal{R}))${\rm;}
    \vspace{1mm}
    
    \item[(b)] $F_k \in \Mat_n (\mathcal{R})$, $k \geq 1$, are off-diagonal matrices{\rm;}
    \vspace{1mm}
    
    \item[(c)] $D_k \in \Mat_n (\mathcal{R})$, $k \geq 1$, are diagonal matrices{\rm;}
    \vspace{1mm}
    
    \item[(d)] $T_k \in \Mat_n(\mathcal{Z}(\mathcal{R}))$, $k = - r, \dots, 1$, and $T_0 \in \Mat_n(\mathcal{R})$ are both diagonal matrices{\rm;}
\end{itemize}
and suppose that 
\begin{itemize}
    \item[(e)]
    the operator $\brackets{k \, \mathbb{I} + \ad_{T_{0}}}: \Mat_n (\mathcal{R}) \to \Mat_n (\mathcal{R})$, $k \geq 1$, is invertible.
\end{itemize}
    \vspace{2mm}

Then the system
\begin{align}
    \label{eq:irrsys}
    \partial_{\lambda} \Phi (\lambda)
    &= A (\lambda) \, \Phi (\lambda)
\end{align}
admits a unique\footnote{Up to a conjugation by an invertible in the sense of Definition \ref{def:invmat} matrix $G \in \Mat_n(\mathcal{R})$.} formal solution near an irregular singular point $\lambda = \infty$ that can be written as
\begin{align}
    \label{eq:formsol}
    &&
    \Phi_{form} (\lambda)
    &= F(\lambda) \, D(\lambda) \,
    \exp \brackets{
    \sum_{k = 1}^{r}
    \tfrac{1}{k} T_{- k} \lambda^{k}
    + \ln (\lambda) \, \, T_0
    }
    &
    \text{as}&
    &
    \lambda 
    &\to 
    \infty.
    &&
\end{align}
\end{prop}
\begin{rem}
Let us make a few comments on the assumptions in the proposition.
\begin{itemize}
    \item Assumption (a) ensures that the off-diagonal matrices $F_k$ can be determined (see eq. \eqref{eq:ass_a_1}). If the entries of the matrix $A_{-r}$ belong to $\mathcal{R}$, then we cannot solve the equation
    \begin{align}
        \left[A_{-r}, F_k\right]_{ij}
        &= F_{k, ij} \, \alpha_j
        - \alpha_i \, F_{k, ij},
    \end{align}
    because the kernel of the operator $\ad_{A_{-r}}: \Mat_n (\mathcal{R}) \to \Mat_n (\mathcal{R})$ is not empty.
    \newline
    Note also that since $A_{-r} \in \Mat_n(\mathcal{Z}(\mathcal{R}))$, we are able to work with Jordan cells, which can be diagonalized using a transformation of the spectral parameter of the form
    \begin{align}
        \lambda 
        \mapsto \lambda^{1/p},
    \end{align}
    as in the commutative case (such an example is given in Proposition \ref{thm:P2_HTW_FN}). In other words, Proposition~\ref{thm:irrsol} is applicable for the fractional Poincar\'e rank\footnote{In this case, $r$ is called \textit{the Katz invariant}.}.

    \item Assumption (d) allows us to define the diagonal matrices $D_k$ (see eq. \eqref{eq:ass_d_1}), since any diagonal matrix with commutative entries commutes with any diagonal matrix with non-commutative entries. For comparison, in the commutative case, any diagonal matrices commute.
    \newline
    Moreover, we can easily differentiate the function $\exp{\brackets{T(\lambda)}}$, where 
    \begin{align}
        T(\lambda)
        &= \sum_{k = 1}^{r} \tfrac{1}{k} 
        T_{- k} \lambda^{k} 
        + \ln (\lambda) \, \, T_0,
    \end{align}
    because $[T_k, T_0] = 0$, $k = - r, \dots, -1$.
    
    \item Thanks to assumption (e), we are able to use a formal monodromy $T_0$ with non-commutative entries. 
\end{itemize}
\end{rem}
\begin{rem}
The statement can be easily reformulated for any irregular singular point different from $\lambda = \infty$.
\end{rem}
\begin{proof}
We are going to show that under the assumptions of the proposition the coefficients $F_k$, $D_k$, and $T_k$ of series \eqref{eq:irrseries} can be found uniquely by the matrix $A(\lambda)$, given in \eqref{eq:irrsys_A}, of system \eqref{eq:irrsys}. In the proof, we will use the following obvious fact:
\begin{lem}
\label{thm:diag_lem}
    For any diagonal matrix $A \in \Mat_n (\mathcal{Z}(\mathcal{R}))$ and any diagonal matrix $B \in \Mat_n (\mathcal{R})$ we have
    \begin{equation}
        [A, B]
        = 0.
    \end{equation}
\end{lem}

\medskip
\textbullet \,\, 
As we mentioned above, one can get
\begin{align}
    &&
    \partial_{\lambda}
    \brackets{e^{T(\lambda)}}
    &= \brackets{\partial_{\lambda} T(\lambda)} \, \, e^{T(\lambda)},
    &
    T(\lambda)
    &= \sum_{k = 1}^{r} \tfrac{1}{k} 
    T_{- k} \lambda^{k} 
    + \ln (\lambda) \, \, T_0,
    &&
\end{align}
because $[T_k, T_0] = 0$, $k = -r, \dots, - 1$ (see assumption (d) and Lemma \ref{thm:diag_lem}).
Thus, substitution of the formal solution
\begin{align}
    \Phi_{form} (\lambda)
    &= F (\lambda) \, D (\lambda)
    e^{T (\lambda)}
\end{align}
into system \eqref{eq:irrsys} leads to the following condition for series $F(\lambda)$, $D (\lambda)$, $T (\lambda)$:
\begin{align}
    \label{eq:mainform}
    \partial_{\lambda} F
    + F \, \brackets{
    \partial_{\lambda} D \, D^{-1}
    + D \,
    \partial_{\lambda} T \, D^{-1}
    }
    &= A \, F,
\end{align}
where we omitted the explicit dependence on $\lambda$. 
Taking diagonal and off-diagonal parts of \eqref{eq:mainform}, we obtain a recurrent system for $F_k$, $D_{k}$, and $T_k$. 
\begin{lem}
The product of the diagonal $A_d$ and off-diagonal $B_{off}$ $n \times n$-matrices is an off-diagonal $n \times n$-matrix $C_{off}$:
\begin{align}
    A_d \, B_{off}
    &= C_{off}.
\end{align}
\end{lem}
Note that $\brackets{F - \mathbb{I}}$ is an off-diagonal matrix. Hence, this system can be written as
\begin{align}
    &&
    \partial_{\lambda} D &= \brackets{A \, F}_d D
    - D \,
    \partial_{\lambda} T,
    &
    \partial_{\lambda} F
    &= A \, F
    - F \, \brackets{A \, F}_d,
    &&
\end{align}
where the subscript $X_d$ means the diagonal part of the matrix $X$. 

\medskip
\textbullet \,\, Consider the off-diagonal part:
\begin{align}
    \label{eq:offdcond}
    \partial_{\lambda} F
    &= A \, F
    - F \, \brackets{A \, F}_d.
\end{align}
Substituting the matrices $A$ and $F$ into \eqref{eq:offdcond} and comparing the coefficients of the corresponding powers of~$\lambda$, we obtain
\begin{align}
    \sum_{l = 0}^k
    \brackets{
    A_{- r + l} \, F_{k - l}
    - \sum_{m = l}^k
    F_{k - m} (A_{- r + l} F_{m - l})_d
    }
    &= 0,
    &
    0 \leq k
    \leq r,&
    \\[2mm]
    \sum_{l = 0}^k
    \brackets{
    A_{- r + l} \, F_{k - l}
    - \sum_{m = l}^k
    F_{k - m} (A_{- r + l} F_{m - l})_d
    }
    + (k - r) F_{k - r}
    &= 0,
    &
    k
    \geq r + 1.&
\end{align}
Note that $(A_{- r} F)_d = A_{- r}$, then for $0 \leq k \leq r$ we have
\begin{align}
    \sum_{l = 0}^k
    \brackets{
    A_{- r + l} \, F_{k - l}
    - \sum_{m = l}^k
    F_{k - m} (A_{- r + l} F_{m - l})_d
    }
    = 0, 
    \\
    A_{- r} \, F_k
    + \sum_{l = 1}^k A_{- r +l} \, F_{k - l}
    - F_k \, A_{- r}
    - \sum_{l = 1}^k \sum_{m = l}^k F_{k - m} (A_{-r + l} F_{m - l})_d
    = 0,
    \\
    [F_k, A_{-r}]
    = \sum_{l = 1}^k \brackets{
    A_{- r + l} \, F_{k - l}
    - \sum_{m = l}^k F_{k - m} (A_{- r + l} F_{m - l})_d
    }.
\end{align}
Similarly for $k \geq r + 1$ one can obtain
\begin{align}
    [F_k, A_{- r}]
    &= \sum_{l = 1}^{k} \brackets{
    A_{- r + l} \, F_{k - l} 
    - \sum_{m = l}^k F_{k - m} \, (A_{-r + l} \, F_{m - l})_d
    }
    + (k - r) F_{k - r},
\end{align}
or, making the change $k \mapsto k - r$, 
\begin{align}
    &&
    [F_{k + r}, A_{- r}]
    &= \sum_{l = 1}^{k + r} \brackets{
    A_{- r + l} \, F_{k + r - l} 
    - \sum_{m = l}^{k + r} F_{k + r - m} \, (A_{- r + l} \, F_{m - l})_d
    }
    + k \, F_{k},
    &
    k 
    &\geq 1.
    &&
\end{align}
Therefore, the coefficients $F_k$ of the off-diagonal part of the formal solution are given by
\begin{align}
    &&
    [F_k, A_{-r}]
    &= \sum_{l = 1}^k \brackets{
    A_{- r + l} \, F_{k - l}
    - \sum_{m = l}^k F_{k - m} (A_{- r + l} F_{m - l})_d
    },
    &
    0 \leq k 
    \leq r,&
    &&
    \\
    &&
    [F_{k + r}, A_{- r}]
    &= \sum_{l = 1}^{k + r} \brackets{
    A_{- r + l} \, F_{k + r - l} 
    - \sum_{m = l}^{k + r} F_{k + r - m} \, (A_{- r + l} \, F_{m - l})_d
    }
    + k \, F_{k},
    &
    k 
    \geq 1.&
    &&
\end{align}
Recall that $A_{-r} = diag(\alpha_1, \dots, \alpha_n) \in \Mat_n(\mathcal{Z}(\mathcal{R}))$ and $\alpha_i \neq \alpha_j$ for any $i \neq j$. Thus, $F_{k, ij}$ can be found from the equations above, since
\begin{align}
    \label{eq:ass_a_1}
    [F_k, A_{-r}]_{ij}
    &= F_{k, ij} \, \alpha_j 
    - \alpha_i \, F_{k, ij}
    = (\alpha_j - \alpha_i) F_{k, ij}
    \neq 0
    ,
\end{align}
where we have used assumption (a).

\medskip
\textbullet \,\, 
Now we proceed to the diagonal part:
\begin{align}
    \label{eq:dcond}
    \partial_{\lambda} D
    &= \brackets{A \, F}_d \, D
    - D \, \partial_{\lambda} T
    ,
\end{align}
that is equivalent to the following recurrent system
\begin{align}
    \label{eq:dcond_1}
    \sum_{l = 0}^k
    \brackets{
    \sum_{m = l}^k 
    (A_{- r + l} \, F_{k - m})_d D_{m - l}
    - D_{k - l} \, T_{- r + l}
    }
    &= 0,
    &
    0 \leq k
    \leq r,&
    \\
    \label{eq:dcond_2}
    \sum_{l = 0}^k
    \sum_{m = l}^k 
    (A_{- r + l} \, F_{k - m})_d D_{m - l}
    - \sum_{l = 0}^r
    D_{k - l} \, T_{- r + l}
    + (k - r) D_{k - r}
    &= 0,
    &
    k
    \geq r + 1.&
\end{align}

Consider $0 \leq k \leq r$ and rewrite the recurrent equation \eqref{eq:dcond_1} in such a way as to gather the terms that do not contain the coefficients $D_l$, $l \geq 1$.
\begin{align}
    \sum_{l = 0}^k
    \sum_{m = l}^k 
    (A_{- r + l} \, F_{k - m})_d D_{m - l}
    - \sum_{l = 0}^k
    D_{k - l} \, T_{- r + l}
    = 0,
    \\
    \sum_{l = 0}^k
    (A_{- r + l} \, F_{k - l})_d D_{0}
    +
    \sum_{l = 0}^k
    \sum_{m = l + 1}^k 
    (A_{- r + l} \, F_{k - m})_d D_{m - l}
    -
    D_{0} \, T_{- r + k}
    -
    \sum_{l = 0}^{k - 1}
    D_{k - l} \, T_{- r + l}
    = 0,
    \\
    \label{eq:ass_d_1}
    T_{- r + k}
    - \sum_{l = 0}^k 
    (A_{- r + l} \, F_{k - l})_d 
    = 
    \sum_{l = 0}^{k - 1}
    \brackets{
    \sum_{m = l + 1}^k 
    (A_{- r + l} \, F_{k - m})_d D_{m - l}
    -
    D_{k - l} \, T_{- r + l}
    }.
\end{align}
Since $T_l \in \Mat_n (\mathcal{Z}(\mathcal{R}))$, $l = -r, \dots, - 1$ (see assumption (d)) we have
\begin{align}
    \label{eq:condT}
    &&
    [D_k, T_l] 
    &= 0,
    &
    k 
    &\geq 1, 
    &
    l 
    &= - r, \dots, - 1.
    &&
\end{align}
Then the coefficients $T_{- r}$, $T_{- r + 1}$, \dots, $T_0$ are defined uniquely and are given by 
\begin{align}
    \label{eq:coefT}
    &&
    T_{- r + k}
    &= \sum_{l = 0}^k 
    (A_{- r + l} \, F_{k - l})_d
    ,
    &
    0 \leq k
    &\leq r.
    &&
\end{align}
In particular, $T_{-r} = A_{-r}$.
We remark that $T_0$ need not commute with $D_k$, $k \geq 1$.

In the case of recurrent system \eqref{eq:dcond_2}, we first simplify it, by using \eqref{eq:condT} and \eqref{eq:coefT}. The resulting system can be written as
\begin{align}
    k D_k
    - D_k \, T_0
    + \sum_{l = 1}^k \sum_{m = 0}^{k - l} (A_l \, F_{k - l - m})_d D_m
    + \sum_{l = - r}^{0} \sum_{m = 0}^{k} (A_l \, F_{k - l - m})_d D_{m}
    = 0
    ,
\end{align}
or, gathering the terms with $D_{k}$,
\begin{align}
    \label{eq:dcond_22}
    &&
    k D_k + [T_0, D_k]
    &= 
    - \sum_{l = 1}^k \sum_{m = 0}^{k - l} (A_{l} \, F_{k - l - m})_d D_m
    - \sum_{l = - r}^0 \sum_{m = 0}^{k - 1} (A_l \, F_{k - l - m})_d D_m,
    &
    k
    &\geq 1,
    &&
\end{align}
where $T_0 = \dsum_{l = 0}^r (A_{- l} F_{l})_d$.
Since the rhs of the expression contains $D_{l}$, $l < k$, we are able to define~$D_{k}$, as the operator $\brackets{k \, \mathbb{I} + [T_0, \, \cdot \,]}$ is invertible by assumption (e).
\end{proof}

In particular, one can find a formal solution near infinity of the systems \ref{eq:P20sys}, \ref{eq:P21sys}, and \ref{eq:P22sys}. Below we assume that the spectral parameter belongs to $\mathbb{C}$.

\begin{prop}
\label{thm:FNsol}
Let $F(\zeta)$, $D(\zeta)$, $T(\zeta)$ satisfy Proposition \ref{thm:irrsol}. Then there exists a unique formal solution, associated with pairs \ref{eq:P20_FNpair}, \ref{eq:P21_FNpair_2}, \ref{eq:P21_FNpair}, and \ref{eq:P22_FNpair}, in the form 
\begin{align}
    \label{eq:formsol_FN}
    &&
    Z_{form} (\zeta)
    &= F(\zeta) \, D(\zeta)
    \,
    \exp \brackets{ 
    \brackets{
    - \tfrac43 \zeta^3
    + \zeta \, x
    } 
    \begin{pmatrix}
    1 & 0 \\ 0 & -1
    \end{pmatrix}
    + \ln (\zeta) \, T_0
    }
    &
    \text{as}&
    &
    \zeta 
    &\to 
    \infty,
    &&
\end{align}
where
\begin{align}
    \text{\ref{eq:P20_FNpair}}:&&
    T_0
    &= \brackets{
    [f, g] - \theta
    } 
    \begin{pmatrix}
    1 & 0 \\ 0 & 1
    \end{pmatrix},
    &
    b
    &\in \mathbb{C};
    &&&&
    \text{\ref{eq:P21_FNpair_2}}:&
    T_0
    &= a
    \begin{pmatrix}
    1 & 0 \\ 0 & 1
    \end{pmatrix};
    \\[2mm]
    \text{\ref{eq:P21_FNpair}}:&&
    T_0
    &= \brackets{
    [f, g] + a
    } 
    \begin{pmatrix}
    1 & 0 \\ 0 & 1
    \end{pmatrix};
    &&&&&&
    \text{\ref{eq:P22_FNpair}}:&
    T_0
    &= 
    \begin{pmatrix}
    0 & 0 \\ 0 & 0
    \end{pmatrix}.
\end{align}
\end{prop}

\begin{proof}
Consider the case of the \ref{eq:P20_FNpair} pair. In other cases computations are similar.

Using Proposition \ref{thm:irrsol}, one can find the off-diagonal part:
\begin{gather}
    \begin{aligned}
    F_1
    &= - \tfrac12 f 
    \begin{pmatrix}
    0 & -1 \\ 1 & 0
    \end{pmatrix}
    ,
    &&&&&
    F_2
    &= - \tfrac18 \brackets{
    2 f^2 - 2 g + x + 2 b
    }
    \begin{pmatrix}
    0 & 1 \\ 1 & 0
    \end{pmatrix}
    ,
    \end{aligned}
    \\[2mm]
    \begin{aligned}
    F_3
    &= \tfrac18 \brackets{
    - f^3 - b f - f b - x f - (\theta - \tfrac12)
    }
    \begin{pmatrix}
    0 & -1 \\ 1 & 0
    \end{pmatrix},
    &&&&&
    &\dots \,.
    \end{aligned}
\end{gather}
The coefficients $T_k$ are
\begin{gather}
    \begin{aligned}
    T_{-3}
    &= - 4 
    \begin{pmatrix}
    1 & 0 \\ 0 & -1
    \end{pmatrix},
    &&&
    T_{-2}
    &= 
    \begin{pmatrix}
        0 & 0 \\ 0 & 0
    \end{pmatrix},
    &&&
    T_{-1}
    &= \brackets{x + 2 b} 
    \begin{pmatrix}
    1 & 0 \\ 0 & - 1
    \end{pmatrix}
    ,
    \end{aligned}
    \\[2mm]
    T_0
    = \brackets{
    [f, g - b] - \theta
    }
    \begin{pmatrix}
    1 & 0 \\ 0 & 1
    \end{pmatrix}
    - 2 \, [D_1, b]
    \begin{pmatrix}
    1 & 0 \\
    0 & - 1
    \end{pmatrix}
    .
\end{gather}
By assumption (d) in Proposition \ref{thm:irrsol}, $ T_{-1} \in \Mat_2 (\mathcal{Z}(\mathcal{R}))$. This is true iff $b \in \mathbb{C}$. Then the coefficient $D_1$ satisfies the equation:
\begin{align}
    D_{1} + \LieBrackets{\LieBrackets{f, g}, D_{1}}
    = \brackets{
    - \tfrac12 f^2 g - \tfrac12 g f^2
    + \tfrac12 g^2 
    - \tfrac12 x g 
    + (\tfrac12 - \theta) f
    + \tfrac18 x^2
    }
    \begin{pmatrix}
    1 & 0 \\ 0 & -1
    \end{pmatrix}
    .
\end{align}
Since the operator $\mathbb{I} + \ad_{\LieBrackets{f, g}}$ is invertible, $D_1$ can be determined. Remaining coefficients $D_k$, $k > 1$, can be found by \eqref{eq:dcond_22}.
\end{proof}

\begin{prop}
\label{thm:JMsol}
Suppose that the series $F(\lambda)$, $D(\lambda)$, $T(\lambda)$ satisfy Proposition \ref{thm:irrsol}. Then there exists a unique formal solution, associated with pairs \ref{eq:P20_JMpair} and \ref{eq:P21_JMpair_2}, in the form 
\begin{align}
    \label{eq:formsol_JM}
    &&
    Y_{form} (\lambda)
    &= F(\lambda) \, D(\lambda)
    \,
    \exp \brackets{ 
    \brackets{
    \tfrac23 \lambda^3
    + \lambda \, x
    } 
    \begin{pmatrix}
    1 & 0 \\ 0 & 0
    \end{pmatrix}
    + \ln (\lambda) \, T_0
    }
    &
    \text{as}&
    &
    \lambda 
    &\to 
    \infty,
    &&
\end{align}
where
\begin{align}
    \text{\ref{eq:P20_JMpair}}:&&
    T_0
    &= [f, g]
    \begin{pmatrix}
    1 & 0 \\ 0 & 0
    \end{pmatrix}
    + (1 - \theta)
    \begin{pmatrix}
    1 & 0 \\ 0 & -1
    \end{pmatrix},
    &
    b
    &\in \mathbb{C};  
    &&&&&&&&
    \\[2mm]
    \text{\ref{eq:P21_JMpair_2}}:&&
    T_0
    &= [f, g]
    \begin{pmatrix}
    0 & 0 \\ 0 & 1
    \end{pmatrix}
    + (1 - a)
    \begin{pmatrix}
    1 & 0 \\ 0 & -1
    \end{pmatrix}.
\end{align}
\end{prop}

\begin{proof}
Consider the case of the \ref{eq:P20_JMpair} pair as an example. The coefficients $F_k$ and $T_{k}$ read as
\begin{gather}
    \begin{aligned}
    F_1
    &= 
    \begin{pmatrix}
    0 & \tfrac12 \brackets{
    2 f^2 - g + x + 2 b
    } 
    \\ - 1 & 0
    \end{pmatrix}
    ,
    &&&&&
    F_2
    &= 
    \begin{pmatrix}
    0 & \tfrac12 \brackets{
    - 2 f^3 + f g - f (x + 2 b) - (\theta - 1)
    } 
    \\ 
    - f & 0
    \end{pmatrix}
    ,
    \end{aligned}
    \\[2mm]
    \begin{aligned}
    F_3
    &= \tfrac12
    \begin{pmatrix}
    0 & 
    2 f^4 - f^2 g - g f^2 + \tfrac12 g^2
    + f^2 (x + 2 b)
    - \tfrac12 g (x + 2 b)
    \\ 
    - 2 f^2 + g & 0
    \end{pmatrix},
    &&&&&
    &\dots \,;
    \end{aligned}
\end{gather}
\begin{gather}
    \begin{aligned}
    T_{-3}
    &= 2
    \begin{pmatrix}
    1 & 0 \\ 0 & 0
    \end{pmatrix},
    &&&
    T_{-2}
    &= 
    \begin{pmatrix}
        0 & 0 \\ 0 & 0
    \end{pmatrix},
    &&&
    T_{-1}
    &= \brackets{x + 2 b} 
    \begin{pmatrix}
    1 & 0 \\ 0 & 0
    \end{pmatrix}
    ,
    \end{aligned}
    \\[2mm]
    T_0
    = \brackets{
    [f, g - 2 b] 
    - 2 \, [D_1, b]
    }
    \begin{pmatrix}
    1 & 0 \\ 0 & 0
    \end{pmatrix}
    + (1 - \theta)
    \begin{pmatrix}
    1 & 0 \\
    0 & - 1
    \end{pmatrix}
    .
\end{gather}
A similar reasoning as in Proposition \ref{thm:FNsol} leads to a restriction for the element $b$, i.e. $b$ should be commutative. Then $D_1$ is defined by the equation
\begin{align}
    D_1
    + \LieBrackets{
    [f, g] 
    \begin{pmatrix}
    1 & 0 \\ 0 & 0
    \end{pmatrix},
    D_1
    }
    = f [f, g]
    \begin{pmatrix}
    1 & 0 \\ 0 & 0
    \end{pmatrix}
    + \brackets{
    g f^2
    - \tfrac12 g^2 
    + \tfrac12 x g 
    - (1 - \theta) f 
    } 
    \begin{pmatrix}
    1 & 0 \\ 0 & -1
    \end{pmatrix}
    .
\end{align}
The coefficients $D_k$ can be found by \eqref{eq:dcond_22}, since the operator $\brackets{k \, \mathbb{I} + \ad_{T_0}}$ is invertible.
\end{proof}

\subsection{Monodromy data}
\label{sec:monsufr}
We will see below that some of the linearizations possess the isomonodromic monodromy data. As a result, the non-abelian monodromy surfaces can be obtained.

Before we start to discuss non-commutative generalizations of the monodromy surfaces related to the JM and FN pairs, let us recall how to construct them in the commutative case.

\subsubsection{The commutative case}
\phantom{}

\textbf{\textbullet \,\, The FN pair.} The commutative FN-type pair is \cite{flaschka1980monodromy}
\begin{align} \label{eq:P2_FNpair}
    \tag*{$\rm{FN}_2$}
    \left\{
    \begin{array}{lcl}
        \partial_{\zeta} Z (\zeta, x)
        &=&
        \left[
        \begin{pmatrix}
        - 4 & 0
        \\[0.9mm]
        0 & 4
        \end{pmatrix} \zeta^2
        + 
        \begin{pmatrix}
        0 & 4 f 
        \\[0.9mm]
        4 f & 0
        \end{pmatrix}
        \zeta
        + 
        \begin{pmatrix}
        2 f^2 + x
        &
        - 2 f^2 + 2 g - x
        \\[0.9mm]
        2 f^2 - 2 g + x
        &
        - 2 f^2 - x
        \end{pmatrix}
        \right.
        \\[3mm]
        && \qquad \qquad \qquad
        \left.
        + 
        \begin{pmatrix}
        - \theta
        & 
        \theta - \tfrac12
        \\[0.9mm]
        \theta - \tfrac12
        & 
        - \theta
        \end{pmatrix}
        \zeta^{-1}
         \right] Z (\zeta, x),
        \\[5mm]
        \partial_{x} Z (\zeta, x)
        &=& \LieBrackets{
        \begin{pmatrix}
        1 & 0
        \\[0.9mm]
        0 & - 1
        \end{pmatrix} 
        \zeta
        +
        \begin{pmatrix}
        0 & - f 
        \\[0.9mm]
        - f & 0
        \end{pmatrix}
         } Z (\zeta, x).
    \end{array}
    \right.
\end{align}
The corresponding formal solution near infinity:
\begin{align}
    \label{eq:formsol_FN2}
    &&
    Z_{form} (\zeta, x)
    &= F(\zeta, x) \, D(\zeta, x)
    \,
    \exp \brackets{ 
    \brackets{
    - \tfrac43 \zeta^3
    + \zeta \, x
    } 
    \begin{pmatrix}
    1 & 0 \\ 0 & -1
    \end{pmatrix}
    - \theta \, \ln (\zeta) \, 
    \begin{pmatrix}
    1 & 0 \\ 0 & 1
    \end{pmatrix}
    }
    &
    \text{as}&
    &
    \zeta 
    &\to 
    \infty.
    &&
\end{align}
The Stokes rays are defined by the equation $Re(\zeta^3) = 0$:
\begin{align}
    &&
    \arg(\zeta)
    &= \frac{\pi}{6} + \frac{\pi \, k}{3},
    &
    k
    &= -1, \dots, 4,
    &&
\end{align}
and, thus, there are six Stokes rays:
\begin{align}
    &&
    l_1
    &:
    &
    \arg(\zeta)
    &= - \frac{\pi}{6},
    &&
    &
    l_2
    &:
    &
    \arg(\zeta)
    &= \frac{\pi}{6},
    &&
    &
    l_3
    &:
    &
    \arg(\zeta)
    &= \frac{\pi}{2},
    &&
    \\[1mm]
    &&
    l_4
    &:
    &
    \arg(\zeta)
    &= \frac{5 \pi}{6},
    &&
    &
    l_5
    &:
    &
    \arg(\zeta)
    &= \frac{7 \pi}{6},
    &&
    &
    l_6
    &:
    &
    \arg(\zeta)
    &= \frac{3 \pi}{2}
    &&
\end{align}
and six Stokes sectors:
\begin{align}
    &&
    l_1 \subset \Omega_1
    &= \brackets{
    - \frac{\pi}{2}, \frac{\pi}{6}
    }
    ,
    &&
    &
    l_2 \subset \Omega_2
    &= \brackets{
    - \frac{\pi}{6}, \frac{\pi}{2}
    }
    ,
    &&
    &
    l_3 \subset \Omega_3
    &= \brackets{
    \frac{\pi}{6}, \frac{5 \pi}{6}
    }
    ,
    &&
    \\[1mm]
    &&
    l_4 \subset \Omega_4
    &= \brackets{
    \frac{\pi}{2}, \frac{7 \pi}{6}
    }
    ,
    &&
    &
    l_5 \subset \Omega_5
    &= \brackets{
    \frac{5 \pi}{6}, \frac{3 \pi}{2}
    }
    ,
    &&
    &
    l_6 \subset \Omega_6
    &= \brackets{
    \frac{7 \pi}{6}, \frac{11 \pi}{6}
    }
    .
    &&
\end{align}
Canonical fundamental solutions $Z_k = Z_k (\zeta, x)$ belong to the corresponding intersections of the Stokes sectors. In our case, we have
\begin{align}
    &&
    Z_1 
    &: &&\Omega_1 \cap \Omega_2 
    = \brackets{- \frac{\pi}{6}, \frac{\pi}{6}},
    &&&
    Z_2 
    &: &&\Omega_2 \cap \Omega_3 
    = \brackets{\frac{\pi}{6}, \frac{\pi}{2}},
    &&&
    Z_3 
    &: &&\Omega_3 \cap \Omega_4 
    = \brackets{\frac{\pi}{2}, \frac{5 \pi}{6}},
    &&
    \\[1mm]
    &&
    Z_4 
    &: &&\Omega_4 \cap \Omega_5 
    = \brackets{\frac{5 \pi}{6}, \frac{7 \pi}{6}},
    &&&
    Z_5 
    &: &&\Omega_5 \cap \Omega_6 
    = \brackets{\frac{7 \pi}{6}, \frac{3 \pi}{2}},
    &&&
    Z_6 
    &: 
    &&\Omega_6 \cap \Omega_1 
    = \brackets{\frac{3 \pi}{2}, \frac{11 \pi}{6}}.
    &&
\end{align}
When passing from one canonical solution $Z_k$ to another $Z_{k + 1}$ around the infinity, the canonical solution $Z_{k}$ is multiplied by the corresponding Stokes matrix $S_k$ on the right:
\begin{align}
    &&
    &&
    Z_{k + 1}
    &= Z_k \, S_k,
    &
    k
    &= 1, \dots, 5,
    &&
    &
    Z_1
    &= Z_6 (\zeta e^{2\pi i}, x) \, S_6.
    &&
\end{align}
By using the condition
\begin{align}
    \label{eq:Stokescond}
    &&
    S_k
    &\equiv Z_k^{-1} \, Z_{k + 1}
    = \lim_{\zeta \to \infty} 
    Z_k^{-1} \, Z_{k + 1}
    ,
    &
    \zeta
    &\in \Omega_k \cap \Omega_{k + 1},
    &&
\end{align}
one can verify that the Stokes matrices are given by
\begin{align}
    S_{1}
    &= 
    \begin{pmatrix}
    1 & 0 \\ s_1 & 1
    \end{pmatrix},
    &
    S_{2}
    &= 
    \begin{pmatrix}
    1 & s_2 \\ 0 & 1
    \end{pmatrix},
    &
    S_{3}
    &= 
    \begin{pmatrix}
    1 & 0 \\ s_3 & 1
    \end{pmatrix},
    &
    S_{4}
    &= 
    \begin{pmatrix}
    1 & s_4 \\ 0 & 1
    \end{pmatrix},
    &
    S_{5}
    &= 
    \begin{pmatrix}
    1 & 0 \\ s_5 & 1
    \end{pmatrix},
    &
    S_{6}
    &= 
    \begin{pmatrix}
    1 & s_6 \\ 0 & 1
    \end{pmatrix}.
\end{align}
\begin{figure}[H]
    \centering
    \scalebox{1.1}{\input{pictures/FN_pair.tex}}
    \caption{The Stokes rays $l_k$, Stokes sectors $\Omega_k$, and Stokes matrices $S_k$ and the canonical solutions $Z_k$, where $k = 1, \dots, 6$.}
    \label{pic:FN_pair}
\end{figure}
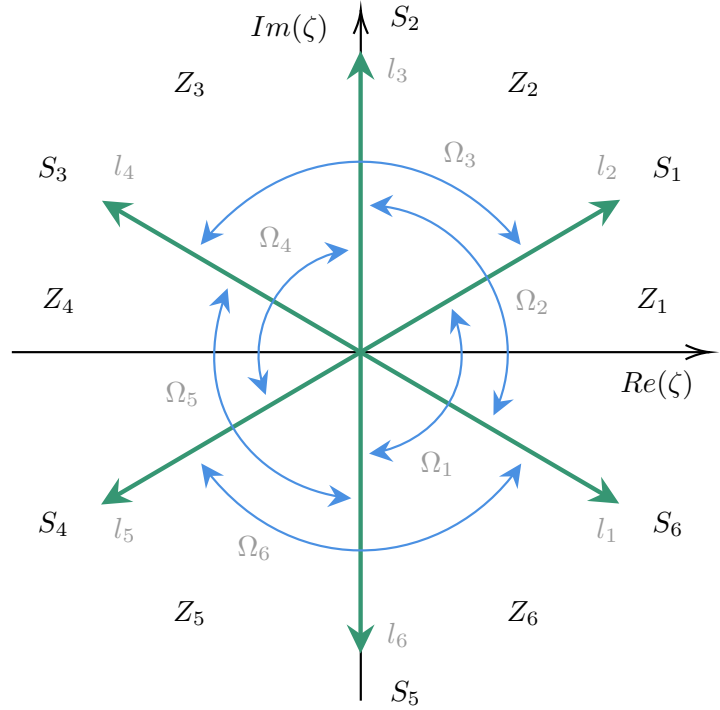
The fundamental solution of the linear system \ref{eq:P2_FNpair} admits the symmetry $Z (e^{\pi i} \zeta, x) = \sigma_1 \, Z(\zeta, x)\,  \sigma_1$, where $\sigma_1 = \small \begin{pmatrix} 0 & 1 \\ 1 & 0 \end{pmatrix}$ is the Pauli matrix\footnote{We will use the standard notation for \href{https://en.wikipedia.org/wiki/Pauli_matrices}{the Pauli matrices}.}, that leads to the relations
\begin{align}
    \label{eq:FNsym}
    &&
    S_{k + 3}
    &= \sigma_1 \, S_k \, \sigma_1,
    &
    s_{k + 3}
    &= s_k,
    &
    k
    &= 1, 2, 3.
    &&
\end{align}

The topological monodromy around zero can be easily found, since the matrix $A_{-1} \in \Mat_2 (\mathbb{C})$ is diagonalizable:
\begin{align}
    &&
    A_{-1}
    &= 
    \begin{pmatrix}
    - \theta & \theta - \tfrac12 \\
    \theta - \tfrac12 & - \theta
    \end{pmatrix}
    = 
    G
    \begin{pmatrix}
    - 2 \theta + \tfrac12 & 0 \\
    0 & - \tfrac12
    \end{pmatrix}
    G^{-1}, 
    &
    G
    &= 
    \begin{pmatrix}
    - 1 & 1 \\ 1 & 1
    \end{pmatrix}
    .
    &&
\end{align}
Therefore, the topological monodromy $M_0$ and the formal monodromy $\Lambda_0$ are
\begin{align}
    &&
    M_0
    &= e^{2 \pi i \, m_0},
    &
    m_0
    &= 
    \begin{pmatrix}
    - 2 \theta + \tfrac12 & 0 \\
    0 & - \tfrac12
    \end{pmatrix},
    &&&
    \Lambda_0
    &= e^{2 \pi i \, T_0},
    &
    T_0
    &= - \theta \, 
    \begin{pmatrix}
    1 & 0 \\ 0 & 1
    \end{pmatrix}
    .
    &&
\end{align}
The monodromy relation\footnote{
In general, the monodromy relation at a singular point $\nu$ is written as \cite{jimbo1981monodromy1}
\begin{align}
    M^{(\nu)}
    &= \brackets{C^{(\nu)}}^{-1} \, e^{2 \pi i \, T_0^{(\nu)}} \, \brackets{
    S_1^{(\nu)} \, \dots S_{k_{\nu}}^{(\nu)}
    }^{-1} \, C^{(\nu)},
\end{align}
where $S_j^{(\nu)}$, $j = 1, \dots, k_{\nu}$, are Stokes multipliers at the point $\nu$, $T_0^{(\nu)}$, $M^{(\nu)}$ are formal monodromy and topological monodromy, respectively, and $C^{(\nu)} \in~SL_n(\mathbb{C})$.
} reads as
\begin{align}
    S_1 \, S_2 \, S_3 \, S_4 \, S_5 \, S_6 \, e^{- 2 \pi i \, T_0}
    &= C \, e^{- 2 \pi i \, m_0} \, C^{-1},
\end{align}
where $C \, C^{-1} = C^{-1} \, C = \mathbb{I}$ and we have fixed $\arg(\zeta) \in (0, 2 \pi]$. Using condition \eqref{eq:FNsym}, it can be rewritten as
\begin{align}
    &&
    \brackets{S_1 \, S_2 \, S_3 \, \sigma_1}^2
    &= C \, e^{- 2 \pi i \, (\theta - \frac12) \sigma_3} \, C^{-1},
    &
    \sigma_3
    &= 
    \begin{pmatrix}
        1 & 0 \\ 0 & -1
    \end{pmatrix}
    ,
    &&
\end{align}
Setting $\tilde \theta = \theta - \tfrac12$, $C = (c_{i,j})$, and $C^{-1} = (\tilde c_{i, j})$ and expanding the following matrix equation\footnote{This equation is defined up to a sign.}
\begin{align}
    &&
    S_1 \, S_2 \, S_3 \, \sigma_1
    &= C \, e^{- \pi i \, 
    \tilde \theta \sigma_3} \, C^{-1},
    &&
\end{align}
we get
\begin{align}   
    \begin{pmatrix}
    s_2
    &
    1 + s_2 \, s_3
    \\
    1 + s_1 \, s_2
    &
    s_1 + s_3 + s_1 \, s_2 \, s_3
    \end{pmatrix}
    &= 
    \begin{pmatrix}
    c_{11} \tilde c_{11} e^{- \pi i \, \tilde \theta}
    + c_{12} \tilde c_{21} e^{\pi i \, \tilde \theta}
    &
    c_{11} \tilde c_{12} e^{- \pi i \, \tilde \theta}
    + c_{12} \tilde c_{22} e^{\pi i \, \tilde \theta}
    \\
    c_{21} \tilde c_{11} e^{- \pi i \, \tilde \theta}
    + c_{22} \tilde c_{21} e^{\pi i \, \tilde \theta}
    &
    c_{21} \tilde c_{12} e^{- \pi i \, \tilde \theta}
    + c_{22} \tilde c_{22} e^{\pi i \, \tilde \theta}
    \end{pmatrix}.
\end{align}
From this system we deduce immediately that the Stokes multipliers are defined by the equation \cite{flaschka1980monodromy}, \cite{van2009moduli}
\begin{align}
    \label{eq:monsurf_FN}
    s_1 \, s_2 \, s_3 
    + s_1 + s_2 + s_3 
    &= 2 \sin (\pi \, \theta).
\end{align}
So, by $\theta$ and any two Stokes multipliers we can determine the third one. As a result, the space of solutions of the $\rm{P}_2$ equation is parameterized by two Stokes multipliers.

\medskip

\textbf{\textbullet \,\, The JM pair.}
In the commutative case of the JM pair \cite{jimbo1981monodromy},
\begin{align} \label{eq:P2_JMpair}
    \tag*{$\rm{JM}_2$}
    \left\{
    \begin{array}{lcl}
        \partial_{\lambda} Y (\lambda, x)
        &=&
        \left[
        \begin{pmatrix}
        2 & 0
        \\[0.9mm]
        0 & 0
        \end{pmatrix}
        \lambda^2
        + 
        \begin{pmatrix}
        0 & - 2 f^2 + g - x
        \\[0.9mm]
        - 2 & 0
        \end{pmatrix}
        \lambda
        \right.
        \\[3mm]
        && \qquad \qquad \quad \,\,
        \left.
        +
        \begin{pmatrix}
        - 2 f^2 + g
        & 
        2 f^3 - f g + x f + \theta - 1
        \\[0.9mm]
        - 2 f 
        & 
        2 f^2 - g + x
        \end{pmatrix}
         \right] Y (\lambda, x),
        \\[5mm]
        \partial_{x} Y (\lambda, x)
        &=& \LieBrackets{
        \begin{pmatrix}
        1 & 0 
        \\[0.9mm] 
        0 & 0
        \end{pmatrix} 
        \lambda
        +
        \begin{pmatrix}
        0 & - f^2 + \frac12 g - \frac12 x
        \\[0.9mm] 
        - 1 & f
        \end{pmatrix}
         } Y (\lambda, x),
    \end{array}
    \right.
\end{align}
the formal solution near the infinity can be written as
\begin{align}
    \label{eq:formsol_JM2}
    &&
    Y_{form} (\lambda, x)
    &= F(\lambda, x) \, D(\lambda, x)
    \,
    \exp \brackets{ 
    \brackets{
    \tfrac23 \lambda^3
    + \lambda \, x
    } 
    \begin{pmatrix}
    1 & 0 \\ 0 & 0
    \end{pmatrix}
    + (1 - \theta) \, \ln (\lambda) \, 
    \begin{pmatrix}
    1 & 0 \\ 0 & -1
    \end{pmatrix}
    }
    &
    \text{as}&
    &
    \lambda 
    &\to 
    \infty.
    &&
\end{align}

The Stokes rays are defined by the equation $Re(\lambda^3) = 0$ and coincide with those written above. Thus, we have the same picture as in Figure \ref{pic:FN_pair}. In contrast to the FN case, the Stokes matrices are
\begin{align}
    S_{1}
    &= 
    \begin{pmatrix}
    1 & s_1 \\ 0 & 1
    \end{pmatrix},
    &
    S_{2}
    &= 
    \begin{pmatrix}
    1 & 0 \\ s_2 & 1
    \end{pmatrix},
    &
    S_{3}
    &= 
    \begin{pmatrix}
    1 & s_3 \\ 0 & 1
    \end{pmatrix},
    &
    S_{4}
    &= 
    \begin{pmatrix}
    1 & 0 \\ s_4 & 1
    \end{pmatrix},
    &
    S_{5}
    &= 
    \begin{pmatrix}
    1 & s_5 \\ 0 & 1
    \end{pmatrix},
    &
    S_{6}
    &= 
    \begin{pmatrix}
    1 & 0 \\ s_6 & 1
    \end{pmatrix}.
\end{align}
Since this pair has only one singular point, the monodromy relation is the following
\begin{align}
    S_1 \, S_2 \, S_3 \, S_4 \, S_5 \, S_6 \, e^{2 \pi i \, \theta \, \sigma_3}
    &= \mathbb{I},
\end{align}
where we have assumed that $\arg(\lambda) \in (0, 2 \pi]$.
This relation can be rewritten as
\begin{align}
    &&
    S_3 \, S_4 \, S_5 \, S_6
    &= \brackets{\alpha^{- \sigma_3} \, S_1 \, S_2}^{-1}
    ,
    &
    \alpha
    &= e^{- 2 \pi i \, \theta},
    &&
\end{align}
or, in explicit form, 
\begin{align*}
    \begin{pmatrix}
    1 
    + s_3 \, s_4 
    + s_3 \, s_6
    + s_5 \, s_6
    + s_3 \, s_4 \, s_5 \, s_6
    & s_3 + s_5 
    + s_3 \, s_4 \, s_5
    \\
    s_4 + s_6 
    + s_4 \, s_5 \, s_6
    & 
    1 + s_4 \, s_5
    \end{pmatrix}
    &= 
    \begin{pmatrix}
    \alpha & - \alpha^{-1} s_1
    \\
    - \alpha \, s_2
    & \alpha^{-1} \brackets{
    1 + s_1 \, s_2
    }
    \end{pmatrix}
    .
\end{align*}
Note that one is able to eliminate from this system the variables $s_1$ and $s_2$. Thus, the resulting equation is
\begin{align}
    \label{eq:premon_JM}
    \alpha \, (1 + s_4 \, s_5)
    &= 1 
    + s_3 \, s_4 
    + s_4 \, s_5 
    + s_3 \, s_4^2 \, s_5 
    + s_3 \, s_6 
    + s_5 \, s_6 
    + 2 s_3 \, s_4 \, s_5 \, s_6
    + s_4 \, s_5^2 \, s_6 
    + s_3 \, s_4^2 \, s_5^2 \, s_6.
\end{align}
Introducing new variables
\begin{align}
    &&
    x_1
    &= 1 + s_3 \, s_4,
    &
    x_2
    &= 1 + s_4 \, s_5,
    &
    x_3
    &= 1 + s_5 \, s_6
    &&
\end{align}
and using the condition
\begin{align}
    \alpha 
    &= 1 + s_3 \, s_4 + s_3 \, s_6 + s_5 \, s_6 + s_3 \, s_4 \, s_5 \, s_6,
\end{align}
the equation \eqref{eq:premon_JM} becomes \cite{van2009moduli}
\begin{align}
    \label{eq:monsurf_JM}
    x_1 \, x_2 \, x_3
    - x_1 - \alpha \, x_2 - x_3
    + (1 + \alpha)
    &= 0,
\end{align}
and parametrizes the space of solutions of the second \Painleve equation.

\begin{rem}
These two cubics, \eqref{eq:monsurf_FN} and \eqref{eq:monsurf_JM}, are connected with each other by the scaling
\begin{align}
    \label{eq:FNtoJM_monsurf}
    &&
    s_1
    &= i \alpha^{\frac12} \, x_1,
    &
    s_2
    &= i \alpha^{- \frac12} \, x_2,
    &
    s_3
    &= i \alpha^{\frac12} \, x_3.
    &&
\end{align}
\end{rem}

\subsubsection{The non-commutative case}

In this subsection, we are going to derive non-commutative generalizations of the monodromy surfaces \eqref{eq:monsurf_JM} and \eqref{eq:monsurf_FN} associated with the JM and FN type pairs, respectively. 

In the non-abelian case, we have non-commutative topological monodromy and formal monodromy matrices and the question is how do they affect the form of the equations defining the monodromy surfaces. 
Regarding the FN pairs, we will follow the idea of their implementation to the monodromy relation suggested in the paper \cite{Bertola2018}. Namely, assume that we have \tomato{two cuts $\mathbb{R}_{\pm}$} on the complex plane and fix the principal value argument of the spectral parameter in such a way as the formal monodromy will appear in the desired order when going around the infinity. Namely, the modified formal solution reads \cite{Bertola2018}
\begin{align}
    \label{eq:formsol_FN_mod}
    &&
    \tilde Z_{form} (\zeta)
    &= F(\zeta) \, D(\zeta)
    \,
    \exp \brackets{ 
    \brackets{
    - \tfrac43 \zeta^3
    + \zeta \, x
    } 
    \begin{pmatrix}
    1 & 0 \\ 0 & -1
    \end{pmatrix}
    + \brackets{
    \ln (\zeta) + 2 \pi i \, \varepsilon
    } \, T_0
    }
    &
    \text{as}&
    &
    \zeta 
    &\to 
    \infty,
    &&
\end{align}
where $\varepsilon = 0$ and $\varepsilon = 1$ in the upper and lower half planes and $\arg(\zeta) \in (0, 2 \pi]$.
\begin{figure}[H]
    \centering
    \scalebox{1.1}{\input{pictures/FN_pair_nc.tex}}
    \caption{The Stokes rays $l_k$, Stokes sectors $\Omega_k$, and Stokes matrices $S_k$ and the canonical solutions $\tilde Z_k$, where $k = 1, \dots, 6$, for the modified fundamental solution $\tilde Z_{form} (\zeta)$.}
    \label{pic:FN_pair_nc}
\end{figure}
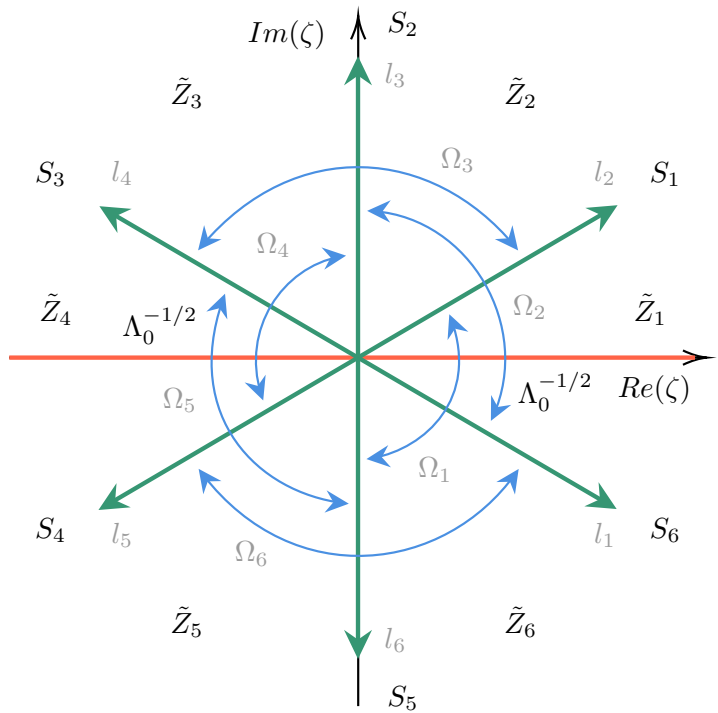

For the \ref{eq:P20_FNpair} pair such a generalization was derived in \cite{Bertola2018}. There are some misprints\footnote{Typos in the signs of the evaluated matrix exponent imply incorrect right-hand-sides of the relations (4.29) and (4.30) and, as a result, lead to another monodromy surfaces.}, that we correct in the following
\begin{prop}
\label{thm:FN20_monsurf}
The Stokes multipliers $s_k \in \mathcal{R}$, $k = 1, 2, 3$, related to the \ref{eq:P20_FNpair} pair, satisfy the equation
\begin{align}
    \label{eq:monsurf_FN_nc}
    s_1 \, s_2 \, s_3 
    + s_1 
    + s_2 
    + s_3
    &= 2 \sin(\pi \, \theta) \, q^{-1},
\end{align}
where $q := e^{\pi i \, [g, f]}\footnote{The function $e$ can be introduced by the Taylor series allowed in $\mathcal{R}$} \in \mathcal{R}$.
\end{prop}
\begin{proof}
\phantom{}

\textbullet \,\, 
\textbf{The Stokes data.}
Note that the Stokes data are similar to the commutative case. Therefore, we have six Stokes rays,
\begin{align}
    &&
    \arg(\zeta)
    &= \frac{\pi}{6} + \frac{\pi \, k}{3},
    &
    k
    &= -1, \dots, 4,
    &&
\end{align}
six Stokes sectors, and the Stokes matrices
\begin{align}
    S_{1}
    &= 
    \begin{pmatrix}
    1 & 0 \\ s_1 & 1
    \end{pmatrix},
    &
    S_{2}
    &= 
    \begin{pmatrix}
    1 & s_2 \\ 0 & 1
    \end{pmatrix},
    &
    S_{3}
    &= 
    \begin{pmatrix}
    1 & 0 \\ s_3 & 1
    \end{pmatrix},
    &
    S_{4}
    &= 
    \begin{pmatrix}
    1 & s_4 \\ 0 & 1
    \end{pmatrix},
    &
    S_{5}
    &= 
    \begin{pmatrix}
    1 & 0 \\ s_5 & 1
    \end{pmatrix},
    &
    S_{6}
    &= 
    \begin{pmatrix}
    1 & s_6 \\ 0 & 1
    \end{pmatrix}
\end{align}
satisfying relation \eqref{eq:FNsym}.

\textbullet \,\, 
\textbf{Topological monodromy and formal monodromy.}
The topological monodromy $M_0$ coincides with the commutative one. The formal monodromy $\Lambda_0$ is (see Proposition \ref{thm:FNsol})
\begin{align}
    &&
    \Lambda_0
    &= e^{2 \pi i \, T_0},
    &
    T_0
    &= \brackets{
    [f,g] - \theta
    } \, 
    \begin{pmatrix}
    1 & 0 \\ 0 & 1
    \end{pmatrix}
    .
    &&
\end{align}
Note that the formal monodromy is an integral of motion of the \ref{eq:P20sys} system, i.e. $\partial_x(T_0) = 0$.

\textbullet \,\,
\textbf{The monodromy relation.}
Let $C \in \Mat_n (\mathcal{R})$ be an invertible matrix and $C^{-1}$ be its inverse, i.e. 
$$C \, C^{-1} = C^{-1} \, C = \mathbb{I}.$$ To obtain the square root form of the monodromy relation, we require \tomato{two cuts $\mathbb{R}_{\pm}$} (they are marked in \tomato{red} in Figure \ref{pic:FN_pair_nc}) and fix the principal value argument as $\arg(\zeta) \in (0, 2 \pi]$. Then the formal solution can be modified in a such way (see eq. \eqref{eq:formsol_FN_mod}) as the monodromy relation reads 
\begin{align}
    S_1 \, S_2 \, S_3 
    \, e^{- \pi i \, T_0} \, 
    S_4 \, S_5 \, S_6 
    \, e^{- \pi i \, T_0}
    &= C \, M_0^{-1} \, C^{-1}.
\end{align}
Note that
\begin{lem}
    For any $a \in \mathcal{R}$ and any matrix $B \in \Mat_n{(\mathcal{Z}(\mathcal{R}))}$, we have
    \begin{align}
        [a \, \mathbb{I}, B]
        &= 0.
    \end{align}
\end{lem}
Then, using symmetry condition \eqref{eq:FNsym} and the lemma, this relation can be rewritten as
\begin{align}
    &&
    \brackets{
    S_1 \, S_2 \, S_3 \, Q \, \sigma_1
    }^2
    &= C \, e^{- 2 \pi i \,
    \brackets{\theta - \frac12} \sigma_3} \, C^{-1},
    &
    Q
    := q \, \mathbb{I}.
    &&
\end{align}
Setting $\tilde \theta = \theta - \tfrac12, C = (c_{i,j})$, and $C^{-1} = (\tilde{c}_{i, j})$, where
\begin{align}
    \begin{pmatrix}
        c_{11} \, \tilde c_{11}
        + c_{12} \, \tilde c_{21}
        & 
        c_{11} \, \tilde c_{12}
        + c_{12} \, \tilde c_{22}
        \\
        c_{21} \, \tilde c_{11}
        + c_{22} \, \tilde c_{21}
        & 
        c_{21} \, \tilde c_{12}
        + c_{22} \, \tilde c_{22}
    \end{pmatrix}
    &= 
    \begin{pmatrix}
        \tilde c_{11} \, c_{11}
        + \tilde c_{12} \, c_{21}
        & 
        \tilde c_{11} \, c_{12}
        + \tilde c_{12} \, c_{22}
        \\
        \tilde c_{21} \, c_{11}
        + \tilde c_{22} \, c_{21}
        & 
        \tilde c_{21} \, c_{12}
        + \tilde c_{22} \, c_{22}
    \end{pmatrix}
    = 
    \begin{pmatrix}
        1 & 0 \\ 0 & 1
    \end{pmatrix},
\end{align}
and expanding the following matrix equation
\begin{align}
    S_1 \, S_2 \, S_3 \, Q \, \sigma_1
    &= C \, e^{- \pi i \, 
    \tilde{\theta} \sigma_3} \, C^{-1},
\end{align}
we get
\begin{align}   
    \begin{pmatrix}
    s_2 \, q 
    &
    (1 + s_2 \, s_3) \, q
    \\
    (1 + s_1 \, s_2) \, q
    &
    (s_1 + s_3 
    + s_1 \, s_2 \, s_3) \, q
    \end{pmatrix}
    &= 
    \begin{pmatrix}
    c_{11} \tilde c_{11} e^{- \pi i \, \tilde \theta}
    + c_{12} \tilde c_{21} e^{\pi i \, \tilde \theta}
    &
    c_{11} \tilde c_{12} e^{- \pi i \, \tilde \theta}
    + c_{12} \tilde c_{22} e^{\pi i \, \tilde \theta}
    \\
    c_{21} \tilde c_{11} e^{- \pi i \, \tilde \theta}
    + c_{22} \tilde c_{21} e^{\pi i \, \tilde \theta}
    &
    c_{21} \tilde c_{12} e^{- \pi i \, \tilde \theta}
    + c_{22} \tilde c_{22} e^{\pi i \, \tilde \theta}
    \end{pmatrix}
    .
\end{align}
Since the condition $C \, C^{-1} = C^{-1} \, C = \mathbb{I}$ holds, the $\tau$-transposed conditions for matrix entries hold also. Thus, one can obtain 
\begin{align}
    s_2 \, q 
    + (s_1 + s_3 + s_1 \, s_2 \, s_3) \, q
    &= \brackets{c_{11} \tilde c_{11} + c_{21} \tilde c_{12}}
    e^{- \pi i \, \tilde \theta}
    + \brackets{c_{12} \tilde c_{21} + c_{22} \tilde c_{22}}
    e^{\pi i \, \tilde \theta}
    \\
    &= \tau \brackets{\brackets{C^{-1} \, C}_{11}} \, e^{- \pi i \, \tilde \theta}
    + \tau \brackets{\brackets{C^{-1} \, C}_{22}} \, e^{\pi i \, \tilde \theta}
    = 2 \cos ({\pi \, \tilde \theta})
    = 2 \sin \brackets{\pi \, \theta}.
\end{align}
\end{proof}

Although for the pairs \ref{eq:P21_FNpair_2} and \ref{eq:P22_FNpair} the formal monodromy is also an integral of motion for the systems \ref{eq:P21sys} and \ref{eq:P22sys}, respectively, the topological monodromy is not like that, because of the term $[f, g]$. Furthermore, besides the fact that the matrix $A_{-1}$ can be reduced to the diagonal form by the same gauge transformation as in the commutative case, the topological monodromy is not an element of $\Mat_2 (\mathcal{Z}(\mathcal{R}))$, which contradicts assumption (a) in Proposition \ref{thm:irrsol}.

Considering the JM type pairs, the formal monodromy is a conserved value for the \ref{eq:P20sys} system only. So, a non-ableian generalization of \eqref{eq:monsurf_JM} is given in the following statement.
\begin{prop}
The Stokes multipliers $x_k \in \mathcal{R}$, $k = 1, 2, 3$, related to the \ref{eq:P20_JMpair} pair, satisfy the condition
\begin{align}
    \label{eq:monsurf_JM_nc}
    x_1 \, x_2 \, x_3
    - x_1 
    - x_2 \, (\alpha \, q^2)
    - x_3
    + (1 + \alpha) \, q^2
    &= 0,
\end{align}
where 
$q := e^{\pi i \, [g, f]} \in \mathcal{R}$ 
and 
$\alpha := e^{- 2 \pi i \, \theta} \in \mathbb{C}$.
\end{prop}
\begin{proof}
\phantom{}

\textbullet \,\, 
\textbf{The Stokes data.}
Note that the Stokes data are the same as in the commutative case. In particular, the Stokes matrices are
\begin{align}
    S_{1}
    &= 
    \begin{pmatrix}
    1 & s_1 \\ 0 & 1
    \end{pmatrix},
    &
    S_{2}
    &= 
    \begin{pmatrix}
    1 & 0 \\ s_2 & 1
    \end{pmatrix},
    &
    S_{3}
    &= 
    \begin{pmatrix}
    1 & s_3 \\ 0 & 1
    \end{pmatrix},
    &
    S_{4}
    &= 
    \begin{pmatrix}
    1 & 0 \\ s_4 & 1
    \end{pmatrix},
    &
    S_{5}
    &= 
    \begin{pmatrix}
    1 & s_5 \\ 0 & 1
    \end{pmatrix},
    &
    S_{6}
    &= 
    \begin{pmatrix}
    1 & 0 \\ s_6 & 1
    \end{pmatrix},
\end{align}
where $s_k \in \mathcal{R}$, $k = 1, \dots, 6$.

\textbullet \,\, 
\textbf{Topological monodromy and formal monodromy.}
The topological monodromy $M_0$ is just the unit matrix, since this pair has only one singular point. According to Proposition \ref{thm:JMsol}, the formal monodromy $\Lambda_0$ is
\begin{align}
    &&
    \Lambda_0
    &= e^{2 \pi i \, T_0},
    &
    T_0
    &= [f, g]
    \begin{pmatrix}
    1 & 0 \\ 0 & 0
    \end{pmatrix}
    + (1 - \theta)
    \begin{pmatrix}
    1 & 0 \\ 0 & -1
    \end{pmatrix}.
    &&
\end{align}

\textbullet \,\,
\textbf{The monodromy relation.}
Let us fix the principal value argument as $\arg(\lambda) \in \left(0, 2 \pi\right]$. Then the monodromy relation is
\begin{align}
    S_1 \, S_2 \, S_3 
    \, S_4 \, S_5 \, S_6 \, e^{- 2 \pi i \, T_0}
    &= \mathbb{I}
    ,
\end{align}
or, equivalently,
\begin{align}
    S_3 \, S_4 \, S_5 \, S_6
    &= \brackets{
    e^{- 2 \pi i \, T_0}
    \, S_1 \, S_2
    }^{-1}
    .
\end{align}
Setting $q = e^{\pi i \, [g, f]}$, $\alpha = e^{- 2 \pi i \, \theta}$ and expanding the matrix equation above, we get the system
\begin{align*}
    \begin{pmatrix}
    1
    + s_3 \, s_4
    + s_3 \, s_6
    + s_5 \, s_6
    + s_3 \, s_4 \, s_5 \, s_6
    & 
    s_3
    + s_5
    + s_3 \, s_4 \, s_5
    \\
    s_4
    + s_6
    + s_4 \, s_5 \, s_6
    & 
    1 + s_4 \, s_5
    \end{pmatrix}
    &= 
    \begin{pmatrix}
    \alpha \, q^2
    & 
    - \alpha^{-1} \, s_1
    \\
    - \alpha \, s_2 \, q^2
    & 
    \alpha^{-1} \, (1 + s_2 \, s_1)
    \end{pmatrix}
    ,
\end{align*}
from which one is able to eliminate the variables $s_1$ and $s_2$. Indeed, on the one hand, we have
\begin{align}
    \alpha \,
    (1 + s_4 \, s_5)
    &= 
    1 + s_2 \, s_1
    &
    \overset{\text{by the $\tau$-action}}{\Longleftrightarrow}
    &&
    \alpha \,
    (1 + s_5 \, s_4)
    &= 
    1 + s_1 \, s_2
    \\
    \label{eq:ncmonJM_22cond}
    &&
    \overset{\phantom{\text{by the $T$-action}}}{\Longleftrightarrow}
    &&
    (1 + s_5 \, s_4) (\alpha \, q^2)
    &= 
    (1 + s_1 \, s_2) \, q^2.
\end{align}
On the other hand, $(1 + s_1 \, s_2) \, q^2 = q^2 + s_1 \, s_2 \, q^2 = q^2 + \brackets{- \alpha^{-1} s_1} \, \brackets{- \alpha \, s_2 \, q^2}$ and, thus, the rhs of \eqref{eq:ncmonJM_22cond} can be rewritten as follows:
\begin{align}
    (1 + s_5 \, s_4) (\alpha \, q^2)
    = q^2 
    + \brackets{
    s_3
    + s_5
    + s_3 \, s_4 \, s_5
    } \, 
    \brackets{
    s_4
    + s_6
    + s_4 \, s_5 \, s_6
    },
    \\[2mm]
    \label{eq:premon_JM_nc}
    \begin{aligned}
    (1 + s_5 \, s_4) (\alpha \, q^2)
    = q^2 
    + s_3 \, s_4 
    + s_3 \, s_6 
    + s_5 \, s_4 
    + s_5 \, s_6 
    + s_3 \, s_4 \, s_5 \, s_4 
    \\[1mm]
    + \, 2 s_3 \, s_4 \, s_5 \, s_6
    + s_5 \, s_4 \, s_5 \, s_6 
    + s_3 \, s_4 \, s_5 \, s_4 \, s_5 \, s_6
    .
    \end{aligned}
\end{align}
Introducing the variables
\begin{align}
    &&
    x_1
    &= 1 + s_3 \, s_4,
    &
    x_2
    &= 1 + s_5 \, s_4,
    &
    x_3
    &= 1 + s_5 \, s_6
    &&
\end{align}
and using the condition
\begin{align}
    \alpha \, q^2
    &= 1 
    + s_3 \, s_4 
    + s_3 \, s_6 
    + s_5 \, s_6 
    + s_3 \, s_4 \, s_5 \, s_6,
\end{align}
equation \eqref{eq:premon_JM_nc} becomes
\begin{align}
    x_1 \, x_2 \, x_3
    - x_1 
    - x_2 \, (\alpha \, q^2)
    - x_3
    + (1 + \alpha) \, q^2
    &= 0.
\end{align}
\end{proof}

\begin{rem}
One can verify that the transformation \eqref{eq:FNtoJM_monsurf} cannot be generalized to the non-commutative case, since, in general, the element $q$ does not commute with the Stokes multipliers. In fact, this issue does not concern with the modification of the formal solution $Z_{form} (\zeta)$. Regarding the same cuts as in the FN case and fixing the principal value argument as $\arg(\lambda) \in (0, 2 \pi]$, we arrive at the monodromy relation for the JM pair of the form
\begin{align}
    &&
    S_3 \, Q \, S_4 \, S_5 \, S_6
    &= \brackets{
    e^{2 \pi i \, \sigma_3} \, Q \, S_1 \, S_2
    }^{-1},
    &
    Q
    &= \small
    \begin{pmatrix}
        q & 0 \\ 0 & 1
    \end{pmatrix}.
    &&
\end{align}
It leads to the equation
\begin{equation}
    (x_1 - 1 + q) \, x_2 \, x_3 
    - x_1 - \alpha \, x_2 - q \, x_3 
    + (2 + \alpha - q)
    = 0,
\end{equation}
that is a particular case of \eqref{eq:monsurf_JM}, since, after a change of the variables, one can rewrite it as
\begin{equation}
    x_1 \, x_2 \, x_3 
    - x_1 - x_2 \, (\alpha q^2) 
    - x_3 
    + (1 + \alpha)
    = 0.
\end{equation}
We also note that the non-abelian monodromy surfaces crucially depend on the principal value argument.
\end{rem}

\begin{rem}
Considering the matrix algebra $\Mat_n(\mathbb{C})$ instead of $\mathcal{R}$, the Stokes multipliers have $3 \, n^2$ arbitrary parameters with one condition for $n^2$ parameters. Therefore, the number of the arbitrary parameters is equal to $2 \, n^2$ and parameterizes the space of solutions of the matrix \ref{eq:P20sys} system.
\end{rem}

\subsection{Poisson structure}
\label{sec:PoisStr}

Going towards the papers \cite{mazzocco2012confluence}, \cite{chekhov2017painleve}, we would like to consider a natural Poisson structure related to non-commutative monodromy surfaces. Recall that commutative monodromy surfaces can be regarded as a family of affine cubics. The cubic surface has a volume form defined by the Poincar\'e residue formulas. This 2-form is holomorphic on the non-singular part of the affine cubic and defines the Poisson brackets on the surface in the natural way. In order to generalize this structure to the non-abelian case, we will start with the non-commutative Poisson brackets.

Since obtained monodromy manifolds can be written in a polynomial form as well\footnote{It is easy to rewrite relation \eqref{eq:monsurf_FN} in the polynomial form.}, we prefer to work with an associative free algebra instead of the ring $\mathcal{R}$. Namely, consider $\mathcal{A} = \mathbb{C} \angleBrackets{x_1, x_2, x_3, q}$ that is an unital associative free algebra over a field $\mathbb{C}$ formed by four generators $x_1$, $x_2$, $x_3$, and $q$.
\begin{defn}
Let $D : \mathcal{A} \to \mathcal{A}$ be an $\mathbb{C}$-linear map satisfying the proprieties
    \begin{enumerate}
        \item $D(f \, g) = D(f) \, g + f \, D(g)$ for any $f$, $g \in \mathcal{A}$,
        \vspace{1mm}
        
        \item $D(\alpha) = 0$ for any $\alpha \in \mathbb{C}$,
        \vspace{1mm}
        
        \item $D(q) = 0$.
    \end{enumerate}
    Such a map we will call \textit{a derivation} of the algebra $\mathcal{A}$. 
\end{defn}

In our case, we have an extra information given by a relation of the form
\begin{equation}
    \varphi (x_1, x_2, x_3, q)
    = 0,
\end{equation}
where $\varphi = \varphi (x_1, x_2, x_3, q)$ is polynomial in all its variables and $\deg \varphi (x_1, x_2, x_3) \leq 3$. Thus, we will introduce Poisson brackets on $\quot{\mathcal{A}}{\angleBrackets{\varphi}}$\footnote{For the shortness, we will omit the ideal.} as follows. 

\begin{defn}
\label{def:PoissonBr}
    Let $\PoissonBrackets{ \, \cdot \, , \, \cdot \, }: \mathcal{A} \times \mathcal{A} \to \mathcal{A}$ be an $\mathbb{C}$-linear map defined as
    \begin{enumerate}
        \item $\PoissonBrackets{x_i, x_j} = \varepsilon_{ijk} \, \partial_k \varphi$,
        where $\varepsilon_{ijk}$ is the antisymmetric symbol and $\partial_k = D_{x_k}$,
        \vspace{1mm}

        \item $\PoissonBrackets{x_i, q} = 0$ for any $x_i \in \mathcal{A}$,
        \vspace{1mm}

        \item $\PoissonBrackets{x_i, x_j} = - \PoissonBrackets{x_j, x_i}$ for any $x_i$, $x_j \in \mathcal{A}$,
    \end{enumerate}
    and satisfying \textit{the left Leibniz rule}
    $
        \PoissonBrackets{f \, g, h}
        = f \, \PoissonBrackets{g, h}
        + g \, \PoissonBrackets{f, h}
    $
    for any $f$, $g$, $h \in \mathcal{A}$.
\end{defn}

\begin{prop}
\label{thm:jacobi}
The map given in Definition {\rm\ref{def:PoissonBr}} satisfies the Jacobi identity
\begin{equation}
    \PoissonBrackets{
    \PoissonBrackets{f, g},
    h
    }
    + 
    \PoissonBrackets{
    \PoissonBrackets{g, h},
    f
    }
    +
    \PoissonBrackets{
    \PoissonBrackets{h, f},
    g
    }
    = 0,
\end{equation}
for any $f$, $g$, $h \in \mathcal{A}$.
\end{prop}
\begin{proof}
    It is enough to prove the Jacobi identity for the generators of the algebra with non-vanishing Poisson brackets, i.e. the identity
\begin{equation}
    \PoissonBrackets{
    \PoissonBrackets{x_1, x_2},
    x_3
    }
    + 
    \PoissonBrackets{
    \PoissonBrackets{x_2, x_3},
    x_1
    }
    +
    \PoissonBrackets{
    \PoissonBrackets{x_3, x_1},
    x_2
    }
    = 0,
\end{equation}
where $\varphi = \varphi (x_1, x_2, x_3, q)$ is a polynomial in all its variables and $\deg \varphi (x_1, x_2, x_3) \leq 3$. That can be done by a straightforward computation.
\end{proof}

\begin{rem}
    Using the $\tau$-action, one can modify Definition \ref{def:PoissonBr} to introduce the right Leibniz rule.
\end{rem}

Thanks to Proposition \ref{thm:jacobi}, the defined brackets are Poisson. We will use them in the further study of the non-abelian volume forms related to the non-abelian cubics. 

\section{Discussion}
Proposition \ref{thm:irrsol} does not allow to work with non-commutative coeffitient $A_{-r}$. Its generalization to a fully-non-commutative case will be helpful to consider a wider class of systems of the form \eqref{eq:irrsys}. We expect that one should change the form of the formal solution near a singular point, avoiding the exponential part. On the other hand, this proposition is still useful to construct non-abelian monodromy surfaces for other types of the non-commutative \Painleve equations, as the isomonodromic Lax pairs are presented for them (see, e.g., \cite{Kawakami_2015}, \cite{Bertola2018}). In addition, this non-commutative cubics can be applied for determining a discrete dynamics that will help to derive non-abelian analogs for the discrete Painlev\'e equations. 

One more interesting question is related to the canonical dynamics on the space of monodromy data (see, e.g., \cite{paul2023dynamics})\footnote{The author is so grateful to Vladimir Rubtsov, who pointed to this paper and this problem.}. These dynamics can be derived from the wild dynamics on the \Painleve foliation \cite{klimes2016wild}. In order to do this, it is necessary to construct a correct definition of log-canonical coordinates on the space of monodromy data. Note that, in the non-abelian setting, $x_1$, $x_2$, $x_3$, and $q$ are generators of the free unital associative algebra $\mathcal{A}$ and, thus, it is not clear how to define the log-canonical symplectic structure. However, in the commutative case, the Poisson bracket defines the Poison tensor. If it is non-singular, then its inverse determines the symplectic structure. Probably, we can introduce a non-abelian symplectic structure in a similar way, by using an appropriate generalization of the $\log$-function to the non-commutative case.

\appendix
\section{FN-type pairs}
\label{ap:FNpairs}

\subsection{Case \texorpdfstring{\ref{eq:P20sys}}{P20}}
\begin{align} \label{eq:P20_FNpair}
    \tag*{$\rm{FN}_2^0$}
    \left\{
    \begin{array}{lcl}
        \partial_{\zeta} Z (\zeta, x)
        &=&
        \left[
        \begin{pmatrix}
        - 4 & 0
        \\[0.9mm]
        0 & 4
        \end{pmatrix} \zeta^2
        + 
        \begin{pmatrix}
        0 & 4 f 
        \\[0.9mm]
        4 f & 0
        \end{pmatrix}
        \zeta
        + 
        \begin{pmatrix}
        2 f^2 + x + 2 b
        &
        - 2 f^2 + 2 g - x - 2 b
        \\[0.9mm]
        2 f^2 - 2 g + x + 2 b
        &
        - 2 f^2 - x - 2 b
        \end{pmatrix}
        \right.
        \\[3mm]
        && \qquad
        \left.
        + 
        \begin{pmatrix}
        - \theta
        & 
        \theta - \tfrac12
        \\[0.9mm]
        \theta - \tfrac12
        & 
        - \theta
        \end{pmatrix}
        \zeta^{-1}
         \right] Z (\zeta, x),
        \\[5mm]
        \partial_{x} Z (\zeta, x)
        &=& \LieBrackets{
        \begin{pmatrix}
        1 & 0
        \\[0.9mm]
        0 & - 1
        \end{pmatrix} 
        \zeta
        +
        \begin{pmatrix}
        0 & - f 
        \\[0.9mm]
        - f & 0
        \end{pmatrix}
         } Z (\zeta, x).
    \end{array}
    \right.
\end{align}

\subsection{Case \texorpdfstring{\ref{eq:P21sys}}{P21}}
\begin{align} \label{eq:P21_FNpair_2}
    \tag*{$\rm{FN}_2^1$}
    \left\{
    \begin{array}{lcl}
        \partial_{\zeta} Z (\zeta, x)
        &=&
        \left[
        \begin{pmatrix}
        - 4 & 0
        \\[0.9mm]
        0 & 4
        \end{pmatrix} \zeta^2
        + 
        \begin{pmatrix}
        0 & 4 f 
        \\[0.9mm]
        4 f & 0
        \end{pmatrix}
        \zeta
        + 
        \begin{pmatrix}
        2 f^2 + x & - 2 f^2 + 2 g - x
        \\[0.9mm]
        2 f^2 - 2 g + x & - 2 f^2 - x
        \end{pmatrix}
        \right.
        \\[3mm]
        && \qquad
        \left.
        +
        \begin{pmatrix}
         [g, f] - a 
         & 
         - [g, f] + a - \frac12
         \\[0.9mm]
        - [g, f] + a - \frac12
        & 
        [g, f] - a
        \end{pmatrix}
        \zeta^{-1}
         \right] Z (\zeta, x),
        \\[5mm]
        \partial_{x} Z (\zeta, x)
        &=& \LieBrackets{
        \begin{pmatrix}
        1 & 0
        \\[0.9mm]
        0 & - 1
        \end{pmatrix} 
        \zeta
        +
        \begin{pmatrix}
        0 & - f 
        \\[0.9mm]
        - f & 0
        \end{pmatrix}
         } Z (\zeta, x).
    \end{array}
    \right.
\end{align}
\begin{align} \label{eq:P21_FNpair}
    \tag*{$\rm{FN'}_2^1$}
    \left\{
    \begin{array}{lcl}
        \partial_{\zeta} Z (\zeta, x)
        &=&
        \left[
        \begin{pmatrix}
        - 4 & 0
        \\[0.9mm]
        0 & 4
        \end{pmatrix} \zeta^2
        + 
        \begin{pmatrix}
        0 & 4 f 
        \\[0.9mm]
        4 f & 0
        \end{pmatrix}
        \zeta
        + 
        \begin{pmatrix}
        2 f^2 + x & - 2 f^2 + 2 g - x
        \\[0.9mm]
        2 f^2 - 2 g + x & - 2 f^2 - x
        \end{pmatrix}
        \right.
        \\[3mm]
        && \qquad
        \left.
        +
        \begin{pmatrix}
         a 
         & 
         a - \frac12
         \\[0.9mm]
        a - \frac12
        & 
        a
        \end{pmatrix}
        \zeta^{-1}
         \right] Z (\zeta, x),
        \\[5mm]
        \partial_{x} Z (\zeta, x)
        &=& \LieBrackets{
        \begin{pmatrix}
        1 & 0
        \\[0.9mm]
        0 & - 1
        \end{pmatrix} 
        \zeta
        +
        \begin{pmatrix}
        f & - f 
        \\[0.9mm]
        - f & f
        \end{pmatrix}
         } Z (\zeta, x).
    \end{array}
    \right.
\end{align}

\subsection{Case \texorpdfstring{\ref{eq:P22sys}}{P22}}
\begin{align} \label{eq:P22_FNpair}
    \tag*{$\rm{FN}_2^2$}
    \left\{
    \begin{array}{lcl}
        \partial_{\zeta} Z (\zeta, x)
        &=&
        \left[
        \begin{pmatrix}
        - 4 & 0
        \\[0.9mm]
        0 & 4
        \end{pmatrix} \zeta^2
        + 
        \begin{pmatrix}
        0 & 4 f 
        \\[0.9mm]
        4 f & 0
        \end{pmatrix}
        \zeta
        + 
        \begin{pmatrix}
        2 f^2 + x
        & 
        - 2 f^2 + 2 g - x - \frac43 b
        \\[0.9mm]
        2 f^2 - 2 g + x + \frac43 b
        &
        - 2 f^2 - x
        \end{pmatrix}
        \right.
        \\[3mm]
        && \qquad
        \left.
        + \left[ g - \frac23 b, f \right]
        \begin{pmatrix}
        1 & - 1
        \\[0.9mm] 
        - 1 & 1
        \end{pmatrix}
        \zeta^{-1}
        +
        \brackets{
        \frac43 b f + a - \tfrac12
        }
        \begin{pmatrix}
        0 & 1 \\ 1 & 0
        \end{pmatrix}
        \zeta^{-1}
        \right.
        \\[3mm]
        && \qquad \qquad
        \left.
        +
        \brackets{\frac23 b f^2 - \frac13 b g + \frac13 x b + \frac49 b^2}
        \begin{pmatrix}
        1 & -1
        \\[0.9mm] 
        1 & -1
        \end{pmatrix}
        \zeta^{-2}
         \right] Z (\zeta, x),
        \\[5mm]
        \partial_{x} Z (\zeta, x)
        &=& \LieBrackets{
        \begin{pmatrix}
        1 & 0 
        \\[0.9mm] 
        0 & - 1
        \end{pmatrix}
        \zeta
        + 
        \begin{pmatrix}
        f & - f
        \\[0.9mm]
        - f & f
        \end{pmatrix}
        + \tfrac16 
        \begin{pmatrix}
        b & -b
        \\[0.9mm] 
        b & - b
        \end{pmatrix}
        \zeta^{-1}
         } Z (\zeta, x).
    \end{array}
    \right.
\end{align}

%% file: pictures/FN_pair.tex
\tikzset{every picture/.style={line width=0.75pt}} 

\begin{tikzpicture}[x=0.75pt,y=0.75pt,yscale=-1,xscale=1]

\draw    (299.67,380.86) -- (299.67,65.95) ;
\draw [shift={(299.67,63.95)}, rotate = 90] [color={rgb, 255:red, 0; green, 0; blue, 0 }  ][line width=0.75]    (10.93,-3.29) .. controls (6.95,-1.4) and (3.31,-0.3) .. (0,0) .. controls (3.31,0.3) and (6.95,1.4) .. (10.93,3.29)   ;
\draw    (139.67,220.95) -- (457.5,220.95) ;
\draw [shift={(459.5,220.95)}, rotate = 180] [color={rgb, 255:red, 0; green, 0; blue, 0 }  ][line width=0.75]    (10.93,-3.29) .. controls (6.95,-1.4) and (3.31,-0.3) .. (0,0) .. controls (3.31,0.3) and (6.95,1.4) .. (10.93,3.29)   ;
\draw [color={rgb, 255:red, 54; green, 150; blue, 115 }  ,draw opacity=1 ][line width=1.5]    (299.58,220.95) -- (415.15,153.04) ;
\draw [shift={(418.6,151.01)}, rotate = 149.56] [fill={rgb, 255:red, 54; green, 150; blue, 115 }  ,fill opacity=1 ][line width=0.08]  [draw opacity=0] (13.4,-6.43) -- (0,0) -- (13.4,6.44) -- (8.9,0) -- cycle    ;
\draw [color={rgb, 255:red, 54; green, 150; blue, 115 }  ,draw opacity=1 ][line width=1.5]    (299.58,220.95) -- (184.45,153.49) ;
\draw [shift={(181,151.47)}, rotate = 30.37] [fill={rgb, 255:red, 54; green, 150; blue, 115 }  ,fill opacity=1 ][line width=0.08]  [draw opacity=0] (13.4,-6.43) -- (0,0) -- (13.4,6.44) -- (8.9,0) -- cycle    ;
\draw [color={rgb, 255:red, 54; green, 150; blue, 115 }  ,draw opacity=1 ][line width=1.5]    (414.72,288.4) -- (299.58,220.95) ;
\draw [shift={(418.17,290.43)}, rotate = 210.37] [fill={rgb, 255:red, 54; green, 150; blue, 115 }  ,fill opacity=1 ][line width=0.08]  [draw opacity=0] (13.4,-6.43) -- (0,0) -- (13.4,6.44) -- (8.9,0) -- cycle    ;
\draw [color={rgb, 255:red, 54; green, 150; blue, 115 }  ,draw opacity=1 ][line width=1.5]    (184.02,288.86) -- (299.58,220.95) ;
\draw [shift={(180.57,290.88)}, rotate = 329.56] [fill={rgb, 255:red, 54; green, 150; blue, 115 }  ,fill opacity=1 ][line width=0.08]  [draw opacity=0] (13.4,-6.43) -- (0,0) -- (13.4,6.44) -- (8.9,0) -- cycle    ;
\draw [color={rgb, 255:red, 54; green, 150; blue, 115 }  ,draw opacity=1 ][line width=1.5]    (299.58,220.95) -- (299.58,86.64) ;
\draw [shift={(299.58,82.64)}, rotate = 90] [fill={rgb, 255:red, 54; green, 150; blue, 115 }  ,fill opacity=1 ][line width=0.08]  [draw opacity=0] (13.4,-6.43) -- (0,0) -- (13.4,6.44) -- (8.9,0) -- cycle    ;
\draw [color={rgb, 255:red, 54; green, 150; blue, 115 }  ,draw opacity=1 ][line width=1.5]    (299.58,355.26) -- (299.58,220.95) ;
\draw [shift={(299.58,359.26)}, rotate = 270] [fill={rgb, 255:red, 54; green, 150; blue, 115 }  ,fill opacity=1 ][line width=0.08]  [draw opacity=0] (13.4,-6.43) -- (0,0) -- (13.4,6.44) -- (8.9,0) -- cycle    ;
\draw [color={rgb, 255:red, 74; green, 144; blue, 226 }  ,draw opacity=1 ]   (306.52,267.02) .. controls (330.55,262.53) and (355.12,239.58) .. (342.49,203.83) ;
\draw [shift={(341.44,201.06)}, rotate = 68.2] [fill={rgb, 255:red, 74; green, 144; blue, 226 }  ,fill opacity=1 ][line width=0.08]  [draw opacity=0] (10.72,-5.15) -- (0,0) -- (10.72,5.15) -- (7.12,0) -- cycle    ;
\draw [shift={(303.44,267.5)}, rotate = 353.06] [fill={rgb, 255:red, 74; green, 144; blue, 226 }  ,fill opacity=1 ][line width=0.08]  [draw opacity=0] (10.72,-5.15) -- (0,0) -- (10.72,5.15) -- (7.12,0) -- cycle    ;
\draw [color={rgb, 255:red, 74; green, 144; blue, 226 }  ,draw opacity=1 ]   (255.14,237.73) .. controls (245.94,207.87) and (264.89,179.66) .. (291.22,174.62) ;
\draw [shift={(294.11,174.17)}, rotate = 172.93] [fill={rgb, 255:red, 74; green, 144; blue, 226 }  ,fill opacity=1 ][line width=0.08]  [draw opacity=0] (10.72,-5.15) -- (0,0) -- (10.72,5.15) -- (7.12,0) -- cycle    ;
\draw [shift={(256.11,240.61)}, rotate = 249.79] [fill={rgb, 255:red, 74; green, 144; blue, 226 }  ,fill opacity=1 ][line width=0.08]  [draw opacity=0] (10.72,-5.15) -- (0,0) -- (10.72,5.15) -- (7.12,0) -- cycle    ;
\draw [color={rgb, 255:red, 74; green, 144; blue, 226 }  ,draw opacity=1 ]   (307.21,153.79) .. controls (350.88,157.31) and (379.57,206.75) .. (361.66,247.5) ;
\draw [shift={(360.5,249.98)}, rotate = 296.16] [fill={rgb, 255:red, 74; green, 144; blue, 226 }  ,fill opacity=1 ][line width=0.08]  [draw opacity=0] (10.72,-5.15) -- (0,0) -- (10.72,5.15) -- (7.12,0) -- cycle    ;
\draw [shift={(303.75,153.61)}, rotate = 1.39] [fill={rgb, 255:red, 74; green, 144; blue, 226 }  ,fill opacity=1 ][line width=0.08]  [draw opacity=0] (10.72,-5.15) -- (0,0) -- (10.72,5.15) -- (7.12,0) -- cycle    ;
\draw [color={rgb, 255:red, 74; green, 144; blue, 226 }  ,draw opacity=1 ]   (237.41,194.95) .. controls (219.39,248.11) and (253.28,282.71) .. (292.92,287.72) ;
\draw [shift={(295.36,287.98)}, rotate = 185.43] [fill={rgb, 255:red, 74; green, 144; blue, 226 }  ,fill opacity=1 ][line width=0.08]  [draw opacity=0] (10.72,-5.15) -- (0,0) -- (10.72,5.15) -- (7.12,0) -- cycle    ;
\draw [shift={(238.61,191.61)}, rotate = 110.54] [fill={rgb, 255:red, 74; green, 144; blue, 226 }  ,fill opacity=1 ][line width=0.08]  [draw opacity=0] (10.72,-5.15) -- (0,0) -- (10.72,5.15) -- (7.12,0) -- cycle    ;
\draw [color={rgb, 255:red, 74; green, 144; blue, 226 }  ,draw opacity=1 ]   (228.22,169.37) .. controls (271.04,119.17) and (333.42,124.05) .. (371.29,169.66) ;
\draw [shift={(373,171.78)}, rotate = 231.62] [fill={rgb, 255:red, 74; green, 144; blue, 226 }  ,fill opacity=1 ][line width=0.08]  [draw opacity=0] (10.72,-5.15) -- (0,0) -- (10.72,5.15) -- (7.12,0) -- cycle    ;
\draw [shift={(226.25,171.73)}, rotate = 309.11] [fill={rgb, 255:red, 74; green, 144; blue, 226 }  ,fill opacity=1 ][line width=0.08]  [draw opacity=0] (10.72,-5.15) -- (0,0) -- (10.72,5.15) -- (7.12,0) -- cycle    ;
\draw [color={rgb, 255:red, 74; green, 144; blue, 226 }  ,draw opacity=1 ]   (228.43,274.56) .. controls (270.51,329.09) and (339.38,319.57) .. (371.79,274.13) ;
\draw [shift={(373.25,272.03)}, rotate = 123.92] [fill={rgb, 255:red, 74; green, 144; blue, 226 }  ,fill opacity=1 ][line width=0.08]  [draw opacity=0] (10.72,-5.15) -- (0,0) -- (10.72,5.15) -- (7.12,0) -- cycle    ;
\draw [shift={(226.5,271.98)}, rotate = 53.88] [fill={rgb, 255:red, 74; green, 144; blue, 226 }  ,fill opacity=1 ][line width=0.08]  [draw opacity=0] (10.72,-5.15) -- (0,0) -- (10.72,5.15) -- (7.12,0) -- cycle    ;

\draw (417.2,227.61) node [anchor=north west][inner sep=0.75pt]    {$Re( \zeta )$};
\draw (247.6,64.41) node [anchor=north west][inner sep=0.75pt]    {$Im( \zeta )$};
\draw (406.3,295.34) node [anchor=north west][inner sep=0.75pt]  [color={rgb, 255:red, 155; green, 155; blue, 155 }  ,opacity=1 ]  {$l_{1}$};
\draw (185,295.34) node [anchor=north west][inner sep=0.75pt]  [color={rgb, 255:red, 155; green, 155; blue, 155 }  ,opacity=1 ]  {$l_{5}$};
\draw (406.33,129.01) node [anchor=north west][inner sep=0.75pt]  [color={rgb, 255:red, 155; green, 155; blue, 155 }  ,opacity=1 ]  {$l_{2}$};
\draw (185,129) node [anchor=north west][inner sep=0.75pt]  [color={rgb, 255:red, 155; green, 155; blue, 155 }  ,opacity=1 ]  {$l_{4}$};
\draw (310.33,84.21) node [anchor=north west][inner sep=0.75pt]  [color={rgb, 255:red, 155; green, 155; blue, 155 }  ,opacity=1 ]  {$l_{3}$};
\draw (310.3,343.51) node [anchor=north west][inner sep=0.75pt]  [color={rgb, 255:red, 155; green, 155; blue, 155 }  ,opacity=1 ]  {$l_{6}$};
\draw (325.59,265.3) node [anchor=north west][inner sep=0.75pt]  [color={rgb, 255:red, 155; green, 155; blue, 155 }  ,opacity=1 ]  {$\Omega _{1}$};
\draw (369.09,190.8) node [anchor=north west][inner sep=0.75pt]  [color={rgb, 255:red, 155; green, 155; blue, 155 }  ,opacity=1 ]  {$\Omega _{2}$};
\draw (336.09,123.3) node [anchor=north west][inner sep=0.75pt]  [color={rgb, 255:red, 155; green, 155; blue, 155 }  ,opacity=1 ]  {$\Omega _{3}$};
\draw (251.59,162.3) node [anchor=north west][inner sep=0.75pt]  [color={rgb, 255:red, 155; green, 155; blue, 155 }  ,opacity=1 ]  {$\Omega _{4}$};
\draw (208.59,233.3) node [anchor=north west][inner sep=0.75pt]  [color={rgb, 255:red, 155; green, 155; blue, 155 }  ,opacity=1 ]  {$\Omega _{5}$};
\draw (241.59,303.8) node [anchor=north west][inner sep=0.75pt]  [color={rgb, 255:red, 155; green, 155; blue, 155 }  ,opacity=1 ]  {$\Omega _{6}$};
\draw (424.59,190.3) node [anchor=north west][inner sep=0.75pt]    {$Z_{1}$};
\draw (365.09,90.8) node [anchor=north west][inner sep=0.75pt]    {$Z_{2}$};
\draw (152.09,190.3) node [anchor=north west][inner sep=0.75pt]    {$Z_{4}$};
\draw (212.09,90.8) node [anchor=north west][inner sep=0.75pt]    {$Z_{3}$};
\draw (212.1,333.51) node [anchor=north west][inner sep=0.75pt]    {$Z_{5}$};
\draw (365.1,333.5) node [anchor=north west][inner sep=0.75pt]    {$Z_{6}$};
\draw (431.9,129.8) node [anchor=north west][inner sep=0.75pt]    {$S_{1}$};
\draw (311.59,60.8) node [anchor=north west][inner sep=0.75pt]    {$S_{2}$};
\draw (150.09,129.8) node [anchor=north west][inner sep=0.75pt]    {$S_{3}$};
\draw (150.09,292.9) node [anchor=north west][inner sep=0.75pt]    {$S_{4}$};
\draw (311.6,370.23) node [anchor=north west][inner sep=0.75pt]    {$S_{5}$};
\draw (431.92,292.9) node [anchor=north west][inner sep=0.75pt]    {$S_{6}$};

\end{tikzpicture}

%% file: pictures/FN_pair_nc.tex
\tikzset{every picture/.style={line width=0.75pt}} 

\begin{tikzpicture}[x=0.75pt,y=0.75pt,yscale=-1,xscale=1]

\draw    (139.67,220.95) -- (457.5,220.95) ;
\draw [shift={(459.5,220.95)}, rotate = 180] [color={rgb, 255:red, 0; green, 0; blue, 0 }  ][line width=0.75]    (10.93,-3.29) .. controls (6.95,-1.4) and (3.31,-0.3) .. (0,0) .. controls (3.31,0.3) and (6.95,1.4) .. (10.93,3.29)   ;
\draw [color={rgb, 255:red, 255; green, 99; blue, 71 }  ,draw opacity=1 ][line width=1.5]    (139.67,220.95) -- (454.08,220.95) ;
\draw    (299.67,380.86) -- (299.67,65.95) ;
\draw [shift={(299.67,63.95)}, rotate = 90] [color={rgb, 255:red, 0; green, 0; blue, 0 }  ][line width=0.75]    (10.93,-3.29) .. controls (6.95,-1.4) and (3.31,-0.3) .. (0,0) .. controls (3.31,0.3) and (6.95,1.4) .. (10.93,3.29)   ;
\draw [color={rgb, 255:red, 54; green, 150; blue, 115 }  ,draw opacity=1 ][line width=1.5]    (299.58,220.95) -- (415.15,153.04) ;
\draw [shift={(418.6,151.01)}, rotate = 149.56] [fill={rgb, 255:red, 54; green, 150; blue, 115 }  ,fill opacity=1 ][line width=0.08]  [draw opacity=0] (13.4,-6.43) -- (0,0) -- (13.4,6.44) -- (8.9,0) -- cycle    ;
\draw [color={rgb, 255:red, 54; green, 150; blue, 115 }  ,draw opacity=1 ][line width=1.5]    (299.58,220.95) -- (184.45,153.49) ;
\draw [shift={(181,151.47)}, rotate = 30.37] [fill={rgb, 255:red, 54; green, 150; blue, 115 }  ,fill opacity=1 ][line width=0.08]  [draw opacity=0] (13.4,-6.43) -- (0,0) -- (13.4,6.44) -- (8.9,0) -- cycle    ;
\draw [color={rgb, 255:red, 54; green, 150; blue, 115 }  ,draw opacity=1 ][line width=1.5]    (414.72,288.4) -- (299.58,220.95) ;
\draw [shift={(418.17,290.43)}, rotate = 210.37] [fill={rgb, 255:red, 54; green, 150; blue, 115 }  ,fill opacity=1 ][line width=0.08]  [draw opacity=0] (13.4,-6.43) -- (0,0) -- (13.4,6.44) -- (8.9,0) -- cycle    ;
\draw [color={rgb, 255:red, 54; green, 150; blue, 115 }  ,draw opacity=1 ][line width=1.5]    (184.02,288.86) -- (299.58,220.95) ;
\draw [shift={(180.57,290.88)}, rotate = 329.56] [fill={rgb, 255:red, 54; green, 150; blue, 115 }  ,fill opacity=1 ][line width=0.08]  [draw opacity=0] (13.4,-6.43) -- (0,0) -- (13.4,6.44) -- (8.9,0) -- cycle    ;
\draw [color={rgb, 255:red, 54; green, 150; blue, 115 }  ,draw opacity=1 ][line width=1.5]    (299.58,220.95) -- (299.58,86.64) ;
\draw [shift={(299.58,82.64)}, rotate = 90] [fill={rgb, 255:red, 54; green, 150; blue, 115 }  ,fill opacity=1 ][line width=0.08]  [draw opacity=0] (13.4,-6.43) -- (0,0) -- (13.4,6.44) -- (8.9,0) -- cycle    ;
\draw [color={rgb, 255:red, 54; green, 150; blue, 115 }  ,draw opacity=1 ][line width=1.5]    (299.58,355.26) -- (299.58,220.95) ;
\draw [shift={(299.58,359.26)}, rotate = 270] [fill={rgb, 255:red, 54; green, 150; blue, 115 }  ,fill opacity=1 ][line width=0.08]  [draw opacity=0] (13.4,-6.43) -- (0,0) -- (13.4,6.44) -- (8.9,0) -- cycle    ;
\draw [color={rgb, 255:red, 74; green, 144; blue, 226 }  ,draw opacity=1 ]   (306.52,267.02) .. controls (330.55,262.53) and (355.12,239.58) .. (342.49,203.83) ;
\draw [shift={(341.44,201.06)}, rotate = 68.2] [fill={rgb, 255:red, 74; green, 144; blue, 226 }  ,fill opacity=1 ][line width=0.08]  [draw opacity=0] (10.72,-5.15) -- (0,0) -- (10.72,5.15) -- (7.12,0) -- cycle    ;
\draw [shift={(303.44,267.5)}, rotate = 353.06] [fill={rgb, 255:red, 74; green, 144; blue, 226 }  ,fill opacity=1 ][line width=0.08]  [draw opacity=0] (10.72,-5.15) -- (0,0) -- (10.72,5.15) -- (7.12,0) -- cycle    ;
\draw [color={rgb, 255:red, 74; green, 144; blue, 226 }  ,draw opacity=1 ]   (255.14,237.73) .. controls (245.94,207.87) and (264.89,179.66) .. (291.22,174.62) ;
\draw [shift={(294.11,174.17)}, rotate = 172.93] [fill={rgb, 255:red, 74; green, 144; blue, 226 }  ,fill opacity=1 ][line width=0.08]  [draw opacity=0] (10.72,-5.15) -- (0,0) -- (10.72,5.15) -- (7.12,0) -- cycle    ;
\draw [shift={(256.11,240.61)}, rotate = 249.79] [fill={rgb, 255:red, 74; green, 144; blue, 226 }  ,fill opacity=1 ][line width=0.08]  [draw opacity=0] (10.72,-5.15) -- (0,0) -- (10.72,5.15) -- (7.12,0) -- cycle    ;
\draw [color={rgb, 255:red, 74; green, 144; blue, 226 }  ,draw opacity=1 ]   (307.21,153.79) .. controls (350.88,157.31) and (379.57,206.75) .. (361.66,247.5) ;
\draw [shift={(360.5,249.98)}, rotate = 296.16] [fill={rgb, 255:red, 74; green, 144; blue, 226 }  ,fill opacity=1 ][line width=0.08]  [draw opacity=0] (10.72,-5.15) -- (0,0) -- (10.72,5.15) -- (7.12,0) -- cycle    ;
\draw [shift={(303.75,153.61)}, rotate = 1.39] [fill={rgb, 255:red, 74; green, 144; blue, 226 }  ,fill opacity=1 ][line width=0.08]  [draw opacity=0] (10.72,-5.15) -- (0,0) -- (10.72,5.15) -- (7.12,0) -- cycle    ;
\draw [color={rgb, 255:red, 74; green, 144; blue, 226 }  ,draw opacity=1 ]   (237.41,194.95) .. controls (219.39,248.11) and (253.28,282.71) .. (292.92,287.72) ;
\draw [shift={(295.36,287.98)}, rotate = 185.43] [fill={rgb, 255:red, 74; green, 144; blue, 226 }  ,fill opacity=1 ][line width=0.08]  [draw opacity=0] (10.72,-5.15) -- (0,0) -- (10.72,5.15) -- (7.12,0) -- cycle    ;
\draw [shift={(238.61,191.61)}, rotate = 110.54] [fill={rgb, 255:red, 74; green, 144; blue, 226 }  ,fill opacity=1 ][line width=0.08]  [draw opacity=0] (10.72,-5.15) -- (0,0) -- (10.72,5.15) -- (7.12,0) -- cycle    ;
\draw [color={rgb, 255:red, 74; green, 144; blue, 226 }  ,draw opacity=1 ]   (228.22,169.37) .. controls (271.04,119.17) and (333.42,124.05) .. (371.29,169.66) ;
\draw [shift={(373,171.78)}, rotate = 231.62] [fill={rgb, 255:red, 74; green, 144; blue, 226 }  ,fill opacity=1 ][line width=0.08]  [draw opacity=0] (10.72,-5.15) -- (0,0) -- (10.72,5.15) -- (7.12,0) -- cycle    ;
\draw [shift={(226.25,171.73)}, rotate = 309.11] [fill={rgb, 255:red, 74; green, 144; blue, 226 }  ,fill opacity=1 ][line width=0.08]  [draw opacity=0] (10.72,-5.15) -- (0,0) -- (10.72,5.15) -- (7.12,0) -- cycle    ;
\draw [color={rgb, 255:red, 74; green, 144; blue, 226 }  ,draw opacity=1 ]   (228.43,274.56) .. controls (270.51,329.09) and (339.38,319.57) .. (371.79,274.13) ;
\draw [shift={(373.25,272.03)}, rotate = 123.92] [fill={rgb, 255:red, 74; green, 144; blue, 226 }  ,fill opacity=1 ][line width=0.08]  [draw opacity=0] (10.72,-5.15) -- (0,0) -- (10.72,5.15) -- (7.12,0) -- cycle    ;
\draw [shift={(226.5,271.98)}, rotate = 53.88] [fill={rgb, 255:red, 74; green, 144; blue, 226 }  ,fill opacity=1 ][line width=0.08]  [draw opacity=0] (10.72,-5.15) -- (0,0) -- (10.72,5.15) -- (7.12,0) -- cycle    ;

\draw (417.2,227.61) node [anchor=north west][inner sep=0.75pt]    {$Re( \zeta )$};
\draw (247.6,64.41) node [anchor=north west][inner sep=0.75pt]    {$Im( \zeta )$};
\draw (406.3,295.34) node [anchor=north west][inner sep=0.75pt]  [color={rgb, 255:red, 155; green, 155; blue, 155 }  ,opacity=1 ]  {$l_{1}$};
\draw (185,295.34) node [anchor=north west][inner sep=0.75pt]  [color={rgb, 255:red, 155; green, 155; blue, 155 }  ,opacity=1 ]  {$l_{5}$};
\draw (406.33,129.01) node [anchor=north west][inner sep=0.75pt]  [color={rgb, 255:red, 155; green, 155; blue, 155 }  ,opacity=1 ]  {$l_{2}$};
\draw (185,129) node [anchor=north west][inner sep=0.75pt]  [color={rgb, 255:red, 155; green, 155; blue, 155 }  ,opacity=1 ]  {$l_{4}$};
\draw (310.33,84.21) node [anchor=north west][inner sep=0.75pt]  [color={rgb, 255:red, 155; green, 155; blue, 155 }  ,opacity=1 ]  {$l_{3}$};
\draw (310.3,343.51) node [anchor=north west][inner sep=0.75pt]  [color={rgb, 255:red, 155; green, 155; blue, 155 }  ,opacity=1 ]  {$l_{6}$};
\draw (325.59,265.3) node [anchor=north west][inner sep=0.75pt]  [color={rgb, 255:red, 155; green, 155; blue, 155 }  ,opacity=1 ]  {$\Omega _{1}$};
\draw (369.09,190.8) node [anchor=north west][inner sep=0.75pt]  [color={rgb, 255:red, 155; green, 155; blue, 155 }  ,opacity=1 ]  {$\Omega _{2}$};
\draw (336.09,123.3) node [anchor=north west][inner sep=0.75pt]  [color={rgb, 255:red, 155; green, 155; blue, 155 }  ,opacity=1 ]  {$\Omega _{3}$};
\draw (251.59,162.3) node [anchor=north west][inner sep=0.75pt]  [color={rgb, 255:red, 155; green, 155; blue, 155 }  ,opacity=1 ]  {$\Omega _{4}$};
\draw (208.59,233.3) node [anchor=north west][inner sep=0.75pt]  [color={rgb, 255:red, 155; green, 155; blue, 155 }  ,opacity=1 ]  {$\Omega _{5}$};
\draw (241.59,303.8) node [anchor=north west][inner sep=0.75pt]  [color={rgb, 255:red, 155; green, 155; blue, 155 }  ,opacity=1 ]  {$\Omega _{6}$};
\draw (424.59,190.3) node [anchor=north west][inner sep=0.75pt]    {$\tilde{Z}_{1}$};
\draw (365.09,90.8) node [anchor=north west][inner sep=0.75pt]    {$\tilde{Z}_{2}$};
\draw (152.09,190.3) node [anchor=north west][inner sep=0.75pt]    {$\tilde{Z}_{4}$};
\draw (212.09,90.8) node [anchor=north west][inner sep=0.75pt]    {$\tilde{Z}_{3}$};
\draw (212.1,333.51) node [anchor=north west][inner sep=0.75pt]    {$\tilde{Z}_{5}$};
\draw (365.1,333.5) node [anchor=north west][inner sep=0.75pt]    {$\tilde{Z}_{6}$};
\draw (431.9,129.8) node [anchor=north west][inner sep=0.75pt]    {$S_{1}$};
\draw (311.59,60.8) node [anchor=north west][inner sep=0.75pt]    {$S_{2}$};
\draw (150.09,129.8) node [anchor=north west][inner sep=0.75pt]    {$S_{3}$};
\draw (150.09,292.9) node [anchor=north west][inner sep=0.75pt]    {$S_{4}$};
\draw (311.6,370.23) node [anchor=north west][inner sep=0.75pt]    {$S_{5}$};
\draw (431.92,292.9) node [anchor=north west][inner sep=0.75pt]    {$S_{6}$};
\draw (190.09,195.8) node [anchor=north west][inner sep=0.75pt]    {$\Lambda _{0}^{-1/2}$};
\draw (371.29,226.3) node [anchor=north west][inner sep=0.75pt]    {$\Lambda _{0}^{-1/2}$};

\end{tikzpicture}